\documentclass[twocolumn]{IEEEtran}

\usepackage{amsmath,epsfig,amssymb,verbatim,amsopn,cite,multirow}
\usepackage{amsthm}
\usepackage{balance}
\usepackage{multirow}
\usepackage[usenames,dvipsnames]{color}
\usepackage[all]{xy}  
\usepackage{url}
\usepackage{amsfonts}
\usepackage{amssymb}
\usepackage{epsfig}
\usepackage{epstopdf}
\usepackage{bm}
\usepackage{balance}
\usepackage{graphicx}
\usepackage{subcaption}
\usepackage{footnote}
\usepackage{algorithm,algorithmic}

\usepackage[margin=19mm,top=19mm,bottom=18mm]{geometry}

  {\proof}{\proofend}
\newtheorem{proposition}{Proposition}

\newcommand{\qa}{{\bf a}}
\newcommand{\qb}{{\bf b}}

\newcommand{\qg}{{\bf g}}

\newcommand{\qq}{{\bf q}}
\newcommand{\qr}{{\bf r}}

\newcommand{\qw}{{\bf w}}
\newcommand{\qx}{{\bf x}}
\newcommand{\qy}{{\bf y}}

\newcommand{\qZ}{{\bf Z}}

\DeclareMathOperator*{\argmax}{arg\,max}

\newcommand{\Sn}{\sigma_n^2}

\newcommand{\FD}{\mathsf{FD}}
\newcommand{\HD}{\mathsf{HD}}
\newcommand{\vFD}{\mathsf{NAFD}}

\newcommand{\Ntx}{N}
\newcommand{\Nrx}{N}

\newcommand{\dl}{\mathtt{dl}}
\newcommand{\ul}{\mathtt{ul}}

\newcommand{\hgmkpd}{\hat{\qg}_{mk'}^{\dl}}
\newcommand{\tgmkd}{\tilde{\qg}_{mk}^{\dl}}
\newcommand{\tgmlu}{\tilde{\qg}_{m\ell}^{\ul}}

\newcommand{\gmkd}{\qg_{mk}^{\dl}}

\newcommand{\hgmkd}{\hat{\qg}_{mk}^{\dl}}

\newcommand{\hgmlu}{\hat{\qg}_{m\ell}^{\ul}}

\newcommand{\gmlu}{\qg_{m\ell}^{\ul}}

\newcommand{\gamdmk}{\gamma_{mk}^{\dl}}
\newcommand{\gamdmkp}{\gamma_{mk'}^{\dl}}
\newcommand{\gamuml}{\gamma_{m\ell}^{\ul}}

\newcommand{\Pbhm}{P_{\mathtt{bh},m}}
\newcommand{\PUfix}{P_{\mathtt{U,fixed}}}
\newcommand{\PbhvFD}{P_{bh}^{\vFD}}
\newcommand{\PtotvFD}{P_{\mathtt{total}}^{\vFD}}
\newcommand{\PbhfulvFD}{P_{bh,full}^{\vFD}}
\newcommand{\PtotbhvFD}{ P_\mathtt{total,fbh}^{\vFD}}

\newcommand{\etamk}{\eta_{mk}}

\newcommand{\vsl}{\varsigma_\ell}

\newcommand{\betamkd}{\beta_{mk}^{\dl}}

\newcommand{\betakldu}{\beta_{k\ell}^{\mathtt{du}}}

\newcommand{\betamlu}{\beta_{mq}^{\ul}}

\newcommand{\alphml}{\alpha_{m\ell}}

\DeclareMathOperator{\aaa}{\mathbf{a}}
\DeclareMathOperator{\FF}{\mathcal{F}}

\DeclareMathOperator{\K}{\mathcal{K}}

\DeclareMathOperator{\OO}{\mathcal{O}}

\DeclareMathOperator{\z}{\mathbf{z}}

\DeclareMathOperator{\HHH}{\mathcal{H}}

\DeclareMathOperator{\SSS}{\mathcal{S}}
\DeclareMathOperator{\CN}{\mathcal{CN}}

\DeclareMathOperator{\AAA}{\mathcal{A}}

\DeclareMathOperator{\bb}{\mathbf{b}}

\DeclareMathOperator{\MM}{\mathcal{M}}

\DeclareMathOperator{\x}{\mathbf{x}}

\DeclareMathOperator{\q}{\mathbf{q}}
\DeclareMathOperator{\g}{\mathbf{g}}

\DeclareMathOperator{\Z}{\mathbf{Z}}

\DeclareMathOperator{\OOmega}{\boldsymbol{\omega}}
\DeclareMathOperator{\ETA}{\boldsymbol{\eta}}

\DeclareMathOperator{\VARSIGMA}{\boldsymbol{\varsigma}}

\DeclareMathOperator{\THeta}{\boldsymbol{\theta}}

\DeclareMathOperator{\ALPHA}{\boldsymbol{\alpha}}

\DeclareMathOperator{\EEEE}{\mathtt{EE}}

\DeclareMathOperator{\vvFD}{\textbf{NAFD} }
\DeclareMathOperator{\RvFD}{\textbf{R-NAFD} }
\DeclareMathOperator{\GvFD}{\textbf{G-NAFD} }
\DeclareMathOperator{\HHD}{\textbf{HD} }
\DeclareMathOperator{\FFD}{\textbf{FD} }
\DeclareMathOperator{\SSSI}{\text{SI}}

\newtheorem{remark}{Remark}

\title{ Network-Assisted Full-Duplex Cell-Free Massive MIMO: Spectral and Energy Efficiencies}

\author{Mohammadali Mohammadi,~\IEEEmembership{Member,~IEEE,} Tung T. Vu,~\IEEEmembership{Member,~IEEE,}\\
Hien Quoc Ngo,~\IEEEmembership{Senior Member,~IEEE,} and  Michail Matthaiou,~\IEEEmembership{Fellow,~IEEE}
\thanks{M. Mohammadi, T. T. Vu, H. Q. Ngo, and M. Matthaiou are with the Centre for Wireless Innovation (CWI), Queen's University Belfast, U.K. Email:\{m.mohammadi, hien.ngo, m.matthaiou\}@qub.ac.uk. (\emph{Corresponding author: Mohammadali Mohammadi.})
}
\thanks{T.~T.~Vu is also with the Department of Electrical Engineering (ISY), Link\"{o}ping University, SE-581 83 Link\"{o}ping, Sweden. Email: thanh.tung.vu@liu.se.}
}

\allowdisplaybreaks

\begin{document}

\bstctlcite{IEEEexample:BSTcontrol}
\maketitle

\begin{abstract} 
We consider network-assisted full-duplex (NAFD) cell-free massive multiple-input multiple-output (CF-mMIMO) systems, where full-duplex (FD) transmission is virtually realized via half-duplex (HD) hardware devices. The HD access points (APs) operating in uplink (UL) mode and those operating in downlink (DL) mode simultaneously serve DL and UL user equipments (UEs) in the same frequency bands. We comprehensively analyze the performance of NAFD CF-mMIMO from both a spectral efficiency (SE) and energy efficiency (EE) perspectives. Specifically, we propose a joint optimization approach that designs the AP mode assignment, power control, and large-scale fading (LSFD) weights to improve the sum SE and EE of NAFD CF-mMIMO systems. We formulate two mixed-integer nonconvex optimization problems of maximizing the sum SE and EE, under realistic power consumption models, and the constraints on minimum individual SE requirements, maximum transmit power at each DL AP and UL UE. The challenging formulated problems are transformed into tractable forms and two novel algorithms are proposed to solve them using successive convex approximation techniques. More importantly, our approach can be applied to jointly optimize power control and LSFD weights for maximizing the sum SE and EE of HD and FD CF-mMIMO systems, which, to date, has not been studied. Numerical results show that: (a) our joint optimization approach significantly outperforms the heuristic approaches in terms of both sum SE and EE; (b) in CF-mMIMO systems, the NAFD scheme can provide approximately $30\%$ SE gains, while achieving a remarkable EE gain of up to $200\%$ compared with the HD and FD schemes.  

\let\thefootnote\relax\footnotetext{The work of M. Mohammadi and M. Matthaiou was supported by a research grant from the Department for the Economy Northern Ireland under the US-Ireland R\&D Partnership Programme. The work of T. T. Vu and H. Q. Ngo was supported by the U.K. Research and Innovation Future Leaders Fellowships under Grant
MR/S017666/1. The work of T.~T.~Vu was also supported in part by ELLIIT and the KAW Foundation. Parts of this paper were presented at IEEE SPAWC 2022~\cite{mohammadSPAWC2022}.  }

\end{abstract}


\section{Introduction}
Cell-free massive multiple-input multiple-output (CF-mMIMO) networks and full-duplex (FD) communications are two technological platforms to address the explosive growth of data demands driven by smartphones, tablets, and other media-hungry devices and are expected to play an important part in fifth-generation (5G) networks and beyond~\cite{Matthaiou:COMMag:2021}. CF-mMIMO represents a scalable and practical implementation of the distributed antenna~\cite{Zhu:JSAC:2011} and network MIMO systems concepts~\cite{Venkatesan:NMIMO}, wherein a large number of access points (APs) are distributed over a coverage area and coordinated by several central processing units (CPUs) to coherently serve many user equipments (UEs) in the same time-frequency resources~\cite{Hien:cellfree}. The application of FD APs in CF-mMIMO enables simultaneous transmission towards downlink (DL) UEs and reception  from uplink (UL) UEs on the same frequency, thus, holds the promise of achieving higher spectral and energy efficiency (SE-EE)~\cite{Nguyen:JSAC:2020}.  

Although the potential gains of FD can be easily foreseen, the key challenge in achieving FD communication is the high transceiver complexity, required to mitigate the self-interference (SI) caused by the signal leakage from the transceiver output to the input~\cite{Duarte:PhD:dis,Sabharwal14JSAC}.  While recent advances in active and passive SI suppression (SIS) techniques~\cite{Korpi:JSAC:2014,Hong:COMMAG:2014} have made the FD transceivers feasible in practice, SIS at the FD AP imposes huge infrastructure costs to the operators as a significant additional processing burden is required. Furthermore, the use of FD transceivers in CF-mMIMO networks gives rise to an additional source of interference, termed as cross-link interference (CLI), i.e., the interference received by the receiving antennas of one AP from the transmitting antennas of another AP, as well as the interference received by a DL UE from other UL UEs. Both SI and CLI subsequently impact the signal processing and resource allocation strategies, making them nontrivial tasks in FD CF-mMIMO compared to the traditional HD CF-mMIMO. More importantly,  CLI results in higher power consumption to achieve the UEs' SE requirements while SIS entails power-hungry hardware at the FD APs~\cite{Zhang2015}. Therefore, the fundamental challenge behind the FD CF-mMIMO realization is how to improve the SE while maintaining the EE at the optimum level. 

The aforementioned challenges have spurred research at the system level design for FD CF-mMIMO systems. More specifically, few efforts have been devoted towards developing functionalities, such as power control, scheduling, and AP selection, to improve the SE and/or EE~\cite{Nguyen:JSAC:2020,Datta:JOP:2022}. Nevertheless, SI and CLI are still major performance bottlenecks, which are significantly detrimental from the power consumption and EE perspective. To deal with the negative impact of CLI, one attractive approach is to revisit and redesign the network architecture. To this end, a novel paradigm for the physical layer of the CF-mMIMO networks, called \emph{network-assisted full duplexing (NAFD)}, has been embodied in~\cite{Wang:TCOM:2020}.  In NAFD CF-mMIMO, the APs can operate in FD (all APs perform DL and UL at the same time over the same frequency), hybrid-duplex (both FD and HD APs exist in the network), or flexible-duplex (APs operate in HD mode), thus, NAFD unifies all duplex modes in the network~\cite{Wang:TCOM:2020,Jiamin:TWC:2021}. The transmission mode of each AP (UL reception and/or DL transmission) can be dynamically designed based on the network conditions and requirements, accordingly both UL and DL services could be supported simultaneously and flexibly.  When NAFD is forced to be established via pure flexible-duplex, SI and related challenges are completely removed. Moreover, the amount of CLI becomes fairly less than that in the FD counterpart. This is due to the fact that only a part of  the APs contribute to DL transmissions, whilst the DL transmissions and remaining APs are dedicated to serving the UL transmissions. Furthermore, AP mode assignment offers a flexible degree-of-freedom to better manage the CLI compared to FD. To reduce the CLI and accordingly enhance the SE-EE of the NAFD CF-mMIMO, joint design of AP mode assignment, UL/DL power control coefficients, and large-scale fading decoding (LSFD) weight design for the UL APs is a promising, but nontrivial direction, which is the main motivation of this paper.

\subsection{Review of Related Literature }
The emulation of the FD operation by spatially separating HD base stations (BSs) in cellular networks, for serving UL and DL UEs, was first introduced in~\cite{Thomsen:WCL:2016}. This scheme, which is known as COMPflex in the literature,  provides advantages over using a FD BS, in terms of SE and EE.  In~\cite{Xin:ACCESS:2017}, a spatial domain duplex, called a bidirectional dynamic network, was proposed for large-scale distributed antenna systems (DAS). In bidirectional dynamic networks, different remote antenna units (RAUs) are used to serve the UL and DL UEs at the same time and over the same frequency band. However, both COMPflex and bidirectional dynamic networks are impractical due to the immense fronthaul signaling and huge computational complexity, which poses performance limitations and system scalability issues.

Inspired by the idea of COMPflex and by leveraging the join processing virtue of CF-mMIMO, NAFD CF-mMIMO networks were proposed in~\cite{Wang:TCOM:2020}. Recently, NAFD CF-mMIMO systems have been investigated from various aspects and/or under different setups~\cite{Wang:TCOM:2020,Xia:TVT:2020,Jiamin:TWC:2021,Xinjiang:TWC:2021,Zhu:COML:2021,Xia:China:2021,Xia:SYSJ:2022}. In particular, the SE of the NAFD CF-mMIMO network has been studied in~\cite{Wang:TCOM:2020} for zero forcing (ZF) and regularized ZF (RZF) precoders in DL and minimum-mean-square-error (MMSE) receiver in UL. An optimization framework for the transceiver design for large-scale DAS with NAFD was developed in~\cite{Xia:TVT:2020}, by maximizing the sum UL and DL SE subject to quality-of-service (QoS) constraints and backhaul constraints. The authors in~\cite{Jiamin:TWC:2021} proposed to utilize a beamforming training scheme to perform interference cancellation at RAUs and coherent decoding at DL UEs. Moreover, closed-form expressions for the UL achievable rates with maximum ratio combining (MRC) and ZF receivers and for the DL achievable rates with maximum ratio transmission (MRT) and ZF beamforming were derived. In~\cite{Xinjiang:TWC:2021}, an optimization framework was established for joint user selection and transceiver design in cell-free with NAFD, where the QoS constraint and fronthaul compression are considered.

Nevertheless, a main underlying assumption in all aforementioned literature~\cite{Wang:TCOM:2020,Xia:TVT:2020,Jiamin:TWC:2021,Xinjiang:TWC:2021} is the fixed mode assignment in NAFD networks, i.e., the UL or DL transmission mode of the RAUs/APs has been already given. This assumption, however, renders the flexible adjustment of UL/DL traffic and efficient resource usage inapplicable. In~\cite{Zhu:COML:2021,Xia:China:2021,Xia:SYSJ:2022}, the problem of RAU mode selection was addressed from different perspectives. In~\cite{Zhu:COML:2021}, RAUs' mode selection for a CF-mMIMO with NAFD was investigated, to maximize the SE of DL and UL UEs, where power budget constraints are considered. This work assumes all the antennas of the same RAU working in the same mode and an RAU is either UL or DL only. In~\cite{Xia:China:2021} and~\cite{Xia:SYSJ:2022}, the authors extended their work in~\cite{Zhu:COML:2021}, to consider hybrid-duplex mode in NAFD CF-mMIMO, where antenna mode assignment at each multi-antenna RAU is investigated. More specifically, a SE maximization problem under power and QoS constraints was studied in~\cite{Xia:China:2021} and a two-stage strategy of antenna mode selection and transceiver design has been proposed. The secrecy SE maximization problem in NAFD CF-mMIMO systems, under the constraints of the operation mode of antennas, UL receivers, and transmit power of APs/UEs was studied in~\cite{Xia:SYSJ:2022},  where artificial noise is applied to interfere with the eavesdropper's reception.

\subsection{Research Gap and Main Contributions}
The ongoing research efforts have mainly focused on two main directions: 1) SE analysis of the NAFD CF-mMIMO under fixed mode assignment at the RAUs (APs), and 2) Enhancing the SE by adopting the AP mode assignment and transceiver design.  However, the main drawback of these studies \cite{Wang:TCOM:2020,Xia:TVT:2020,Jiamin:TWC:2021,Xinjiang:TWC:2021,Zhu:COML:2021,Xia:China:2021,Xia:SYSJ:2022} is that they use the instantaneous channel state information (CSI) for system level designs rather than the statistical CSI. Therefore, all resource allocation  and mode assignment designs must be recomputed quickly once the small-scale fading coefficients are changed. Moreover, these designs must be done at the CPUs, thus, the APs have to send all channel estimates to the CPUs which will cause very large overhead, especially in CF-mMIMO where the numbers of APs and UEs are very large. Finally, they require instantaneous CSI knowledge at the UEs which causes huge resources and overhead in the systems with many UEs~\cite{Hien:cellfree}. Another major issue with the aforementioned optimization frameworks is the sub-optimality, as the original problems are decoupled into sub-problems, which are alternately solved via iterative algorithms. Furthermore, all the pioneering research in the area of NAFD CF-mMIMO has focused on optimization schemes to maximize the SE, while the trade-off between SE and EE has remained an open research problem.  From a green perspective, energy consumption has become a critical concern for planning 5G/6G networks because mobile communication networks may contribute significantly towards the global carbon footprint. To ensure sustainability, 5G networks should operate at low energy consumption levels while still achieving large SE~\cite{Shafi:JSAC:2017}.

\begin{table*}
	\caption{ List of Notation} 
	\vspace{-0.5em}
	\centering 
	\begin{tabular}{|c | c |}
\hline
$K_d (K_u)$ & Number of DL (UL) UEs\\
\hline
$M$ & Number of APs\\
\hline
$N$ & Number of antennas per-AP\\
\hline
$\gmkd$ ($\gmlu$) & Channel vector between DL UE $k$ (UL UE $\ell$) and AP $m$\\
\hline
$\betamkd$ ($\betamlu$) &  Large-scale fading coefficient\\
\hline
$h_{k\ell}$ &Channel gain between the UL UE $\ell$ to the DL UE $k$\\
\hline
$\qZ_{mi}$ &Channel matrix from AP $m$ to AP $i$\\
\hline
$\etamk$  &Power control coefficient at the AP $m$\\
\hline
${\varsigma}_{\ell}$ &Transmit power control coefficient at UL UE $\ell$
\\
\hline
$\rho_t$ & Normalized transmit power of each pilot symbol\\
\hline
$\rho_d$ & Maximum normalized transmit power at each AP\\
\hline
$\rho_u$ & Maximum normalized transmit power at each UL UE\\
\hline
$\tau_c$ & Length of coherence block \\
\hline
$\tau_t$ &Length of pilot sequences  \\
\hline
$a_m$, $b_m$ &AP mode assignment binary variables\\
\hline
$\alphml$ &LSFD weight at AP $m$ corresponding to UL UE $\ell$\\
\hline
$\mathcal{S}_\ul^o$ ($\mathcal{S}_\dl^o$) &Minimum SE required by the $\ell$-th UL UE ($k$-th DL UE)\\
\hline
$P_{\mathtt{cdl},m}$  ($P_{\mathtt{cul},m}$) &Internal power consumption at AP $m$ for DL (UL) transmissions\\
\hline
$\zeta_m$  ($\chi$) &Power amplifier efficiency at AP $m$ (UE)\\
\hline
$P_{\mathtt{bt},m}$ &Traffic-dependent backhaul power\\
\hline
$P_{\mathtt{fdl},m}$ ($P_{\mathtt{ful},m}$) &Fixed power consumption for each DL (UL) backhaul\\
\hline
$P_{\mathtt{D},\ell}$ $(P_{\mathtt{U},k})$ &Fixed power consumption for DL (UL) UE\\
\hline
	\end{tabular}
	\label{Notations}
	\vspace{0.95em}
\end{table*}

The research on NAFD CF-mMIMO is still in its infancy stage and several key challenges must be considered.  While it is well-known that HD/FD CF-mMIMO is energy efficient~\cite{Hien:TGCN:2018,Nguyen:JSAC:2020}, there is still no previous work on characterizing/optimizing the SE and EE of NAFD CF-mMIMO, relying on statistical CSI. To the authors' best knowledge, the SE and EE maximization, that takes into account the effects of imperfect CSI, hardware power consumption, QoS requirement of all UEs, per-AP and UE power control, AP mode assignment, and LSFD weights has not been studied for NAFD CF-mMIMO. The closest work to our research is~\cite{chowdhury2021can}, wherein the performance of a NAFD CF-mMIMO system under dynamic time division duplex (TDD) was analyzed, where the transmission mode at each AP is scheduled so that the sum UL-DL SE is maximized. However, in~\cite{chowdhury2021can}, the AP modes are scheduled by a greedy algorithm, which is not optimal. Moreover, the SE requirements for UL and DL UEs were ignored, and the impact of power control as well as LSFD for UL reception were not investigated.

Motivated by filling the above-mentioned knowledge gap in the literature, we consider an  NAFD CF-mMIMO system with TDD operation, where  HD multi-antenna APs simultaneously serve multiple UL and DL UEs on the same time-frequency resources. Then, the SE and EE of the NAFD CF-mMIMO system are investigated comprehensively. The main contributions of this paper are summarized as follows:

\begin{itemize}
\item We propose a joint optimization approach of AP mode assignment, power control, and LSFD weights to improve the SE and EE of NAFD CF-mMIMO systems. The proposed approach is general and can be applied to both the traditional HD and FD CF-mMIMO systems.
 
\item We formulate two novel optimization problems for SE and EE maximization of the NAFD CF-mMIMO system. The formulated problems are under realistic power consumption models (including power consumption for hardware and backhaul links), individual SE requirements of both UL and DL UEs, per-AP and per-UL UE power constraints. 

\item Two new algorithms are then proposed to solve the challenging formulated mixed-integer non-convex problems. In particular, we transform the formulated problems into more tractable problems with continuous variables only. Then, we solve the problems using successive convex approximation techniques.
 
\item Numerical results show that our joint optimization approach significantly outperforms the heuristic approaches. The simulation results also confirm that the NAFD scheme remarkably improves the performance of the CF-mMIMO in terms of both the SE and EE over the traditional HD and FD schemes. 
\end{itemize}

\subsection{Paper Organization and Notation}
The remainder of this paper is organized as follows. In Section~\ref{sec:Sysmodel}, we describe the NAFD CF-mMIMO model, derive the UL/DL SE expressions and present the power consumption model. The formulation of the SE and EE optimization problem and the derivation of their solutions are provided in Section~\ref{sec:SE} and~\ref{sec:EE}, respectively. In Section~\ref{sec:bench}, several benchmarks are presented. Numerical results and discussions are provided in Section~\ref{Sec:Numer}, while Section~\ref{Sec:conc} concludes the paper.

\textit{Notation:} We use bold upper case letters to denote matrices, and lower case letters to denote vectors. The superscripts $(\cdot)^*$, $(\cdot)^T$ and $(\cdot)^\dag$ stand for the conjugate, transpose, and conjugate-transpose (Hermitian), respectively;  $\mathbf{I}_N$ denotes the $N\times N$ identity matrix. The zero mean circular symmetric complex Gaussian distribution having variance $\sigma^2$ is denoted by $\mathcal{CN}(0,\sigma^2)$. Finally, $\mathbb{E}\{\cdot\}$ denotes the statistical expectation.  Table~\ref{Notations} lists some of the important notations used in this article.
 
\vspace{-0em}
\section{System model}~\label{sec:Sysmodel}
We consider a NAFD CF-mMIMO system  under TDD operation, where $M$ APs serve $K_u$ UL UEs and $K_d$ DL UEs. Each AP is connected to the CPU via a high-capacity backhaul link. Each UE is equipped with one single antenna, while each AP is equipped with $N$ antennas. All APs and UEs are HD devices.  As shown in Fig.~\ref{fig:systemmodel}, the assigned UL and DL APs perform simultaneous DL and UL transmissions over the same frequency band.  Each coherence block includes two phases: UL training for channel estimation and UL-and-DL payload data transmission.

\subsection{Uplink Training for Channel Estimation}
\label{phase:ULforCE}
The channel vector between the $k$-th DL UE ($\ell$-th UL UE) and the $m$-th AP is denoted by $\gmkd\in\mathbb{C}^{\Ntx \times 1}$ ($\gmlu\in\mathbb{C}^{\Nrx \times 1}$), $\forall k \in \K_d\triangleq \{1,\dots,K_d\}, \ell \in \K_u\triangleq \{1,\dots,K_u\}, m \in \MM \triangleq \{1, \dots, M\}$. It is modeled as $\gmkd=\sqrt{\betamkd}\tgmkd,~(\gmlu=\sqrt{\betamlu}\tgmlu) $, where $\betamkd$ ($\betamlu$) is the large-scale fading coefficient and $\tgmkd\in\mathbb{C}^{\Ntx \times 1}$ ($\tgmlu\in\mathbb{C}^{\Ntx \times 1}$) is the small-scale fading vector whose elements are independent and identically distributed (i.i.d.) $\mathcal{CN} (0, 1)$ random variables (RVs). Moreover, the channel gain between the UL UE $\ell$ to the DL UE $k$ is denoted by $h_{k\ell}$. It can be modeled as $h_{k\ell}=(\betakldu)^{1/2}\tilde{h}_{k\ell}$, where $\betakldu$ is the large-scale fading coefficient and $\tilde{h}_{k\ell}$ is a $\mathcal{CN}(0,1)$ RV. Finally, the interference links among the APs are modeled as Rayleigh fading channels. Let $\qZ_{mi}\in \mathbb{C}^{\Nrx\times\Ntx}$, $i\neq m$, be the channel matrix from AP $m$ to AP $i$, $\forall m,i\in\MM$, whose elements are  i.i.d. $\mathcal{CN}(0,\beta_{mi})$ RVs. Here, we set $\Z_{mm} = 0, \forall m$. 

\begin{figure}[t]
	\centering
	\vspace{0em}
	\includegraphics[width=90mm]{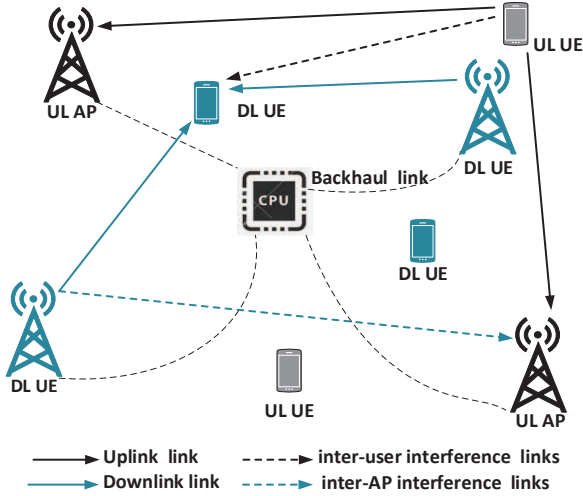}
 	\vspace{-17em}
	\caption{ Illustration of a NAFD CF-mMIMO system with the assigned UL and DL APs along with the received desired and interference signals at a typical DL UE and UL AP.}
	\vspace{-1em}
	\label{fig:systemmodel}
\end{figure}

In each coherence block of length $\tau_c$, all UEs are assumed to transmit their pairwisely orthogonal pilot sequences of length $\tau_t$ to all the APs, which requires $\tau_t\geq K_d + K_u$. At AP $m$, $\g^\dl_{mk}$  and $\g^\ul_{m\ell}$ are estimated by using the received pilot signals and the MMSE estimation technique. By following~\cite{ngo17TWC}, the MMSE estimates $\hgmkd$ and $\hgmlu$ of $\gmkd$  and $\gmlu$ are $\hgmkd \sim \mathcal{CN}(\boldsymbol{0},\gamdmk \mathbf{I}_N)$, and $\hgmlu \sim \mathcal{CN}(\boldsymbol{0},\gamuml \mathbf{I}_N)$, respectively, where $\gamdmk \triangleq
\frac{{\tau_t\rho_t}(\betamkd)^2}
{\tau_t\rho_t
\betamkd+1}, 
\gamuml \triangleq 
\frac{{\tau_t\rho_t}(\betamlu)^2}
{\tau_t\rho_t
\betamlu
+1}$, with $\rho_t$ being the normalized transmit power of each pilot symbol.
\vspace{-0.5em}
\subsection{Downlink-and-Uplink Payload Data Transmission}
In this phase, the APs are able to switch between the UL and DL modes. The decision of which mode is assigned to each AP is optimized to achieve the highest sum SE or total EE of the network as will be discussed in Sections~\ref{sec:SE} and~\ref{sec:EE}, respectively. Note that the AP mode selection is performed on the large-scale fading timescale which changes very slowly with time. The binary variables to indicate the mode assignment for each AP $m$ are defined as
\begin{align}
\label{a}
a_{m} \triangleq
\begin{cases}
  1, & \text{if AP $m$ operates in the DL mode,}\\
  0, & \mbox{otherwise}.
\end{cases}, \forall m
\\
\label{b}
b_{m} \triangleq
\begin{cases}
  1, & \text{if AP $m$ operates in the UL mode,}\\
  0, & \mbox{otherwise}.
\end{cases}, \forall m
\end{align}
Here, we have
\vspace{-1em}
\begin{align}
\label{sumab}
    a_m + b_m = 1, \forall m,
\end{align}
to guarantee that AP $m$ only operates in either the DL or UL mode. 

\subsubsection{Downlink payload data transmission}
By using the local channel estimates, the APs perform maximum-ratio (MR) processing (a.k.a. conjugate beamforming) for the signals transmitted to $K_d$ DL UEs. Our choice of MR beamforming is inspired by the fact that it has low computational complexity and can be implemented in a distributed manner. Thus, most of the processing is done locally at the APs and there is no need to exchange CSI between the APs and CPU~\cite{Hien:cellfree,Emil:TWC:2021:Aging}.  Moreover,  from a green perspective, MR processing  has much less power consumption as compared to the ZF and MMSE processing schemes~\cite{Emil:TWC:2015:EE}.

Let $s_k^\dl$ denote the intended symbol for DL UE $k$. We assume that $s_k^\dl$ is a RV with zero mean and unit variance. The transmitted signal $\qx_{m}^{d}\in\mathbb{C}^{\Ntx\times 1}$ from AP $m$ in the DL mode is generated by first scaling each symbol $s_k^\dl$
with the power control coefficient 
\begin{align}
    \label{theta}
    \theta_{mk} \geq 0, \forall m,k,
\end{align}
and then
multiplying them with the MR precoding vector  $\left(\hgmkd\right)^*$ as 
\begin{align}
 \qx_{m}^{\dl}
= \sqrt{\rho_d}\sum_{k \in \mathcal{K}_d} \theta_{mk} \left(\hgmkd\right)^*
s_{k}^{\dl},   
\end{align}
where $\rho_d$ is the maximum normalized transmit power at each AP. Here, we enforce
\begin{align}
\label{etaa:relation}
      (\theta_{mk} = 0, \forall k,\,\, \text{if}\,\, a_m = 0),  \forall m,
\end{align}
to ensure that if AP $m$ does not operate in the DL mode, all the transmit powers $\rho_d \theta_{mk}^2, \forall k$, at AP $m$ are zero.
Note that AP $m$ is required to meet the average normalized
power constraint, i.e., $\mathbb{E}\left\{\|\qx_{m}^{\dl}\|^2\right\}\leq \rho_d$, which can also be
expressed as the following per-AP power constraint~\cite{Hien:cellfree}
\begin{align}
\label{DL:power:cons}
\sum_{k\in\mathcal{K}_{d}} \gamdmk \theta_{mk}^2 \leq \frac{1}{\Ntx}, \forall m.
\end{align}
The received signal at DL UE $k$ is written as
\begin{align}~\label{eq:ykdl}
y_k^{\dl}
&=
\sqrt{\rho_d}\sum_{m \in \mathcal{M}} \theta_{mk}
\left(\gmkd\right)^T\left(\hgmkd\right)^*
s_{k}^{\dl}
\nonumber\\
&\hspace{2em}+
\sqrt{\rho_d}
\sum_{m \in \mathcal{M}}
\sum_{k'\in\mathcal{K}_d \setminus k} \theta_{mk'}
\left(\gmkd\right)^T\left(\hgmkpd\right)^*
s_{k'}^{\dl}
\nonumber\\
&\hspace{2em}
+
\sum_{\ell\in \mathcal{K}_{u}}h_{k\ell}x_{\ell}^{\ul}+w_{k}^{\dl},
\end{align}
where $w_{k}^{\dl}\sim\mathcal{CN}(0,1)$ is the AWGN at DL UE $k$. We notice that the third term in~\eqref{eq:ykdl}, is the CLI caused by the UL UEs due to concurrent transmissions of DL and UL UEs over the same frequency band, while $x_\ell^\ul$ denotes the transmit signal from the $\ell$-th UL UE. 

\subsubsection{Uplink payload data transmission}
The transmitted signal  from  UL UE $\ell$ is represented by $x_{\ell}^\ul  = \sqrt{\rho_u {\varsigma}_\ell} s_{\ell}^{\ul}$, where $s_{\ell}^\ul$, with $\mathbb{E}\left\{|s_{\ell}^\ul|^2\right\}=1$, and $\rho_u$ denote respectively the transmitted symbol by the $\ell$-th UL UE and  the maximum normalized transmit power at each UL UE, and ${\varsigma}_{\ell}$ is the transmit power control coefficient at UL UE $\ell$ with
\begin{align}
\label{UL:power:cons}
    0\leq {\varsigma}_{\ell} \leq 1, \forall \ell.
\end{align}
The UL APs with $b_m=1, \forall m$, receive a transmit signal from all UL UEs. The received signal $\qy_{m}^{\ul}\in\mathbb{C}^{\Nrx \times 1}$ at AP $m$ in the UL mode can be written as
\begin{align}\label{eq:ymul}
\qy_{m}^{\ul}
&=
\sqrt{\rho_u}\sum_{\ell\in \mathcal{K}_{u}}\sqrt{b_m {\varsigma}_{\ell}}\qg_{m\ell}^{\ul} s_{\ell}^{\ul}
\nonumber\\
&\hspace{2em}
+
\sqrt{\rho_d}\sum_{i\in\mathcal{M}\setminus m}\sum_{k\in \mathcal{K}_d}
\sqrt{b_m } \theta_{ik}
\qZ_{mi}
(\hat{\qg}_{ik}^\dl)^*s_k^\dl
\nonumber\\
&\hspace{2em}
+\sqrt{b_m}\qw_{m}^{\ul},
\end{align}
where $\qw_{m}^{\ul}$ is the $\mathcal{CN}(0,1)$ AWGN vector. We recall that~\eqref{eq:ymul} captures the fact that if AP $m$ does not operate in the UL mode, i.e., $b_m=0$, it does not receive any signal, i.e, $\qy_{m}^{\ul}=\boldsymbol{0}$. 

Then, AP $m$ performs MRC processing (i.e., matched filter) by applying the Hermitian of the (locally obtained) channel estimation vector $\hgmlu$ to the received signal in~\eqref{eq:ymul}. The resulting $(\hgmlu)^\dag\qy_{m}^{\ul}$ is then forwarded to the CPU for signal detection. In order to improve the achievable UL SE, the forwarded signal is further multiplied by the LSFD weight, $\alphml, \forall m,k$. The aggregated received signal for UL UE $\ell, \forall \ell,$ at the CPU can be written as~\cite{Bashar:TWC:2019}
\begin{align}\label{eq:rul}	r_{\ell}^{\ul}=\sum_{m=1}^{M}\alphml(\hgmlu)^\dag\qy_{m}^{\ul}. 
\end{align}
Finally, $s_{\ell}^{\ul}$ is detected from $r_{\ell}^{\ul}$. Without loss of generality, we
assume that
\begin{align}\label{eq:alphml}
|\alphml|^2\leq 1,  \quad \forall \ell, m.
\end{align}

\subsubsection{Downlink SE}
In order to detect $s_{k}^{\dl}$ from the received signal \eqref{eq:ykdl}, the $k$-th DL UE is assumed to rely on the stochastic channel state information. To this end, by applying the use-and-then-forget capacity-bounding technique~\cite{Hien:cellfree}, a closed-form expression for the achievable DL SE (in bits/s/Hz) can be obtained as
\begin{align}~\label{eq:DL:SE}
\mathcal{S}_{\dl,k}^{\vFD} (\qa, \boldsymbol \theta, {\boldsymbol{\varsigma}}) =  \frac{\tau_c-\tau_t}{\tau_c}
\log_2 \left(1  
    + \frac{\Xi_k^2}{\Omega_k}
	\right),
	\end{align}
where $\qa \triangleq \{a_m\}$, $\boldsymbol{\theta}\triangleq \{\theta_{mk}\}$, ${\boldsymbol{\varsigma}} \triangleq \{{\varsigma}_{\ell}\}, \forall m,k,\ell$, and
\begin{align}
    \nonumber
    &\Xi_k (\boldsymbol \theta) \triangleq \Ntx \sqrt{\rho_{d}}
                           \sum_{{m\in\mathcal{M}}}
                            \theta_{mk} \gamdmk,
    \\
    \nonumber
    &\Omega_k (\boldsymbol \theta, {\boldsymbol{\varsigma}}) \triangleq \rho_{d}\Ntx
    \!\sum_{k'\in\mathcal{K}_{d}}\!\sum_{m\in\mathcal{M}}\!\!\!
    \theta_{mk'}^2 \betamkd\gamdmkp + \rho_u\!\!\sum_{\ell\in\mathcal{K}_u} \!\! {\vsl} \betakldu \!+\! 1.
\end{align}
 The detailed derivation of~\eqref{eq:DL:SE} is provided in Appendix~\ref{DL:SE:proof}.
\subsubsection{Uplink SE}
The CPU detects the desired signal $s_{\ell}^{\ul}$ from $	\qr_{\ell}^{\ul}$ in~\eqref{eq:rul}. Since the CPU does not know the instantaneous CSI, it can efficiently use statistical knowledge of the channels when performing the detection. Using again the use-and-then-forget capacity-bounding technique~\cite{Hien:cellfree}, we obtain the achievable UL SE (in bits/s/Hz) of the UL UE $\ell$ as
\begin{align}
    \label{eq:UL:SE}
	\mathcal{S}_{\ul,\ell}^{\vFD} (\qb, \boldsymbol{\varsigma}, \boldsymbol{\theta}, \boldsymbol{\alpha} )
	=\! \frac{\tau_c\!-\!\tau_t}{\tau_c}\log_2
	(1+ \text{SINR}_{\ul,\ell}^{\vFD}(\qb, \boldsymbol{\varsigma}, \boldsymbol{\theta}, \boldsymbol{\alpha} )),
\end{align}
where $\text{SINR}_{\ul,\ell}^{\vFD}(\qb, \boldsymbol{\varsigma}, \boldsymbol{\theta}, \boldsymbol{\alpha} )$ is given in~\eqref{eq:SINRulell} at the top of the next page,
\begin{figure*}
\begin{align}~\label{eq:SINRulell}
\text{SINR}_{\ul,\ell}^{\vFD}(\qb, \boldsymbol{\varsigma}, \boldsymbol{\theta}, \boldsymbol{\alpha} )=
\frac{
	\Nrx \rho_{u} \left(\sum\limits_{\substack{m\in\mathcal{M}}} \sqrt{b_m \varsigma_{\ell}} \alphml \gamuml \right)^2
	}
	{\rho_{u}
		\sum\limits_{\substack{m\in\mathcal{M}}}
		 \sum\limits_{q\in\mathcal{K}_u}\!
		\!
		b_m \varsigma_{q}
		\alphml^2
		\betamlu
		\gamma_{m\ell}^{\ul}
		+
		\rho_{d}\Ntx
		\!\sum\limits_{\substack{m\in\mathcal{M}}}
		\sum\limits_{\substack{i\in\mathcal{M}}}
		\sum\limits_{k\in\mathcal{K}_d}
		b_m \theta_{ik}^2 \alphml^2  \gamuml \beta_{mi} \gamma_{ik}^{\dl}
		+
		\sum\limits_{\substack{m\in\mathcal{M}}}		
		b_m\alphml^2\gamuml}.
  \end{align}
  	\hrulefill
	\vspace{-4mm}
  \end{figure*}
$\qb \triangleq \{b_m\}$, and $\boldsymbol{\alpha} \triangleq\{\alpha_{m\ell}\}, \forall m, k, \ell$. 
Note that, if $a_m = b_m = \alpha_{mk} = 1, \forall m,k,$ \eqref{eq:UL:SE} reduces to the UL SINR of the FD CF-mMIMO system given in~\cite[Eq. (27)]{tung19ICC}.

\subsection{Power Consumption Model}
Let $P_{\ell}$ and $P_{\mathtt{U},\ell}$ be the power consumption for transmitting signals and the required power consumption to run circuit components for the UL transmission at UL UE $\ell$.  Moreover, denote by $P_{\mathtt{D},k}$ the power consumption to run circuit components for the DL transmission at DL UE $k$;
$\Pbhm$ is the power consumed by the backhaul link between the CPU and AP $m$. Therefore, the total power consumption over the considered NAFD CF-mMIMO system is modeled as~\cite{Emil:TWC:2015:EE,Hien:TGCN:2018}
\begin{align}~\label{eq:Ptotal}
P_\mathtt{total}^{\vFD} &= \sum_{\ell\in\mathcal{K}_{u}} (P_{\ell} + P_{\mathtt{U},\ell}) + \sum_{m\in\mathcal{M}} P_m^{\vFD}  \nonumber\\
&\hspace{2em}+ \sum_{k\in\K_d} P_{\mathtt{D},k} + \sum_{m\in\mathcal{M}} \Pbhm^{\vFD},
\end{align}
where $P_m^{\vFD}$ denotes the power consumption at AP $m$ that includes the power consumption of the transceiver chains and the power consumed for the DL or UL transmission. The power consumption $P_m^{\vFD}$ can be modeled as~\cite{Emil:TWC:2015:EE,Hien:TGCN:2018}
\begin{align}~\label{eq:Pm}
\nonumber
&P_m^{\vFD} (\boldsymbol a, \boldsymbol b,\boldsymbol \theta) = 
\nonumber\\
&\hspace{0em}
\begin{cases}
  \frac{1}{\zeta_m}\rho_{d}\Sn\left(N\sum_{k\in\mathcal{K}_{d}} \gamdmk \theta_{mk}^2\right)
+N P_{\mathtt{cdl},m},
\nonumber\\
&\hspace{-4em} \text{if $a_m=1$ }
\\
  N P_{\mathtt{cul},m}, 
& \hspace{-4em} \mbox{if $b_m=1$},
\end{cases}, \forall m
\\
\nonumber
& \overset{\eqref{a}-\eqref{sumab}}{=} a_m\left[\frac{1}{\zeta_m}\rho_{d}\Sn\left(N\sum_{k\in\mathcal{K}_{d}} \gamdmk \theta_{mk}^2 \right) 
+N P_{\mathtt{cdl},m}\right]
\nonumber\\
&\hspace{4em}
+ b_m N P_{\mathtt{cul},m}
\\
& \overset{\eqref{etaa:relation}}{=} \frac{1}{\zeta_m}\rho_{d}\Sn\left(N\sum_{k\in\mathcal{K}_{d}} \gamdmk \theta_{mk}^2 \right) 
+ a_m N P_{\mathtt{cdl},m} \nonumber\\
&\hspace{4em}
+ b_m N P_{\mathtt{cul},m},
\end{align}
where $0<\zeta_m\leq 1 $ is the power amplifier efficiency at the $m$-th AP, $\Sn$ is the noise power; $P_{\mathtt{cdl},m}$ and $P_{\mathtt{cul},m}$ are the internal power required to run the circuit components (e.g., converters, mixers, and filters) related to each antenna of AP $m$ for the DL and UL transmissions, respectively. The power consumption at UL UE $\ell$ is given by
\begin{align}~\label{eq:Pl}
P_{\ell}= \frac{1}{\chi}\rho_u\Sn {\varsigma}_{\ell}, 
\end{align}
where $\chi$ is the power amplifier efficiency at UL UEs.

Let $B$ be the system bandwidth. The backhaul rate between AP $m$ and the CPU is
\begin{align}~\label{eq:Sx}
R_m^{\vFD}(\qx) &=  B \Big( a_m \sum_{\ell\in\mathcal{K}_u} \mathcal{S}_{\ul,\ell}^{\vFD} (\qb, \boldsymbol \varsigma, \boldsymbol \theta, \boldsymbol \alpha)   \nonumber\\
&\hspace{4em}
+ b_m \sum_{k\in\mathcal{K}_d}\mathcal{S}_{\dl,k}^{\vFD} (\qa, \boldsymbol \theta, {\boldsymbol{\varsigma}}) \Big),
\end{align}
where $\qx=\{\qa, \qb, \boldsymbol \varsigma, \boldsymbol \theta, \boldsymbol \alpha\}$. Denote by $P_{\mathtt{fdl},m}$ (resp. $P_{\mathtt{ful},m}$) the fixed power consumption for the DL (resp. UL) transmission of each backhaul, which is traffic-independent and may depend on the distances between the APs and the CPU and the system topology. Then, the power consumption of the backhaul signal load to each AP $m$ is proportional to the backhaul rate as \cite{Hien:TGCN:2018,bashar21TCOM}
\begin{align}~\label{eq:Pbhm}
\Pbhm^{\vFD} = a_m P_{\mathtt{fdl},m} + b_m P_{\mathtt{ful},m} + R_m^{\vFD}(\qx)
P_{\mathtt{bt},m},
\end{align}
 where $P_{\mathtt{bt},m}$ is the traffic-dependent backhaul power (in Watt per bit/s).
By substituting~\eqref{eq:Pm},~\eqref{eq:Pl}, and~\eqref{eq:Pbhm} into~\eqref{eq:Ptotal}, we have
\begin{align}~\label{eq:Ptotal:final}
P_\mathtt{total}^{\vFD} (\x)
&= 
  \sum_{m\in\mathcal{M}} \frac{N\rho_{d}\Sn}{\zeta_m}\left(\sum_{k\in\mathcal{K}_{d}} \gamdmk \theta_{mk}^2 
\right) 
\nonumber\\
&\hspace{2em}
+\sum_{\ell\in\mathcal{K}_{u}} \frac{\rho_u\Sn}{\chi} {\varsigma}_{\ell} +\PUfix + \PbhvFD \nonumber\\
& \hspace{2em} +  \sum_{m\in\mathcal{M}} a_m (N P_{\mathtt{cdl},m} + P_{\mathtt{fdl},m}) 
\nonumber\\
&\hspace{2em}
+ \sum_{m\in\mathcal{M}}  b_m (NP_{\mathtt{cul},m} + P_{\mathtt{ful},m}),
\end{align}
where $\PUfix \triangleq \sum_{k\in\K_u} P_{\mathtt{U},\ell} +  \sum_{k\in\K_d} P_{\mathtt{D},k}$ and 
\begin{align}
    \label{PbhvFD}
    \PbhvFD &\triangleq B \sum_{m\in\mathcal{M}}\Big(b_m\sum_{\ell\in\mathcal{K}_u} \mathcal{S}_{\ul,\ell}^{\vFD} (\qb, \boldsymbol \varsigma, \boldsymbol \theta, \boldsymbol \alpha)   
     \nonumber\\
     &\hspace{2em}+  a_m\sum_{k\in\mathcal{K}_d}\mathcal{S}_{\dl,k}^{\vFD} (\qa, \boldsymbol \theta, {\boldsymbol{\varsigma}})\Big)P_{\mathtt{bt},m},
\end{align}
is the total rate-dependent power consumption in backhaul links. 
\section{Spectral Efficiency Maximization}
\label{sec:SE}
\subsection{Problem Formulation}
In this subsection, we seek to optimize the UL and DL mode assignment vectors $(\aaa,\bb)$, power control coefficients $(\THeta, \VARSIGMA)$, and LSFD weight $\ALPHA$, to maximize the total SE, under the constraints on per-UE SE, transmit power at each AP and UL UE. More precisely,  we formulate an optimization problem as
\begin{subequations}\label{P:SE}
	\begin{align}
		\underset{\qx}{\max}\,\, &
		\SSS ^{\vFD} (\x)
		\\
		\mathrm{s.t.} \,\,
		\nonumber
		& \eqref{a}-\eqref{sumab}, \eqref{theta},
		\eqref{etaa:relation}, \eqref{DL:power:cons}, \eqref{UL:power:cons}, \eqref{eq:alphml} 
		\\
		& \mathcal{S}_{\ul,\ell}^{\vFD} (\qb, \boldsymbol \varsigma, {\boldsymbol{\theta}},\boldsymbol \alpha) \geq \mathcal{S}_\ul^o,~\forall \ell
		\label{UL:QoS:cons}
		\\
		&\mathcal{S}_{\dl,k}^{\vFD} (\qa, \boldsymbol \theta, {\boldsymbol{\varsigma}}) \geq  \mathcal{S}_\dl^o,~\forall k, 
		\label{DL:QoS:cons}
	\end{align}
\end{subequations}
where $\qx$ has been given after~\eqref{eq:Sx}, $\mathcal{S}_\ul^o$ and $\mathcal{S}_\dl^o$ are the minimum SE required by the $\ell$-th UL UE and $k$-th DL UE, respectively, to guarantee the QoS in the network. Moreover, the total SE of a NAFD CF-mMIMO system is defined as
\begin{align}
    \SSS ^{\vFD}(\x) \triangleq \sum_{\ell\in\mathcal{K}_u} \mathcal{S}_{\ul,\ell}^{\vFD} (\qb, \boldsymbol \varsigma, \boldsymbol \theta, \boldsymbol \alpha)   +  \sum_{k\in\mathcal{K}_d}\mathcal{S}_{\dl,k} ^{\vFD}(\qa, \boldsymbol \theta, {\boldsymbol{\varsigma}}).
\end{align} 

For the sake of algorithmic design in later sections, we first transform problem \eqref{P:SE} into a more tractable form as follows: 
\begin{subequations}\label{P:SE:equi}
\begin{align}
\underset{\qx,\qq_{\ul},\qq_{\dl}}{\min}\,\, &
- \sum_{\ell\in\mathcal{K}_u} q_{\ul,\ell}   -  \sum_{k\in\mathcal{K}_d} q_{\dl,k}  \\
\mathrm{s.t.} \,\,
\nonumber
& \eqref{a}-\eqref{sumab}, \eqref{theta},
		\eqref{etaa:relation}, \eqref{DL:power:cons}, \eqref{UL:power:cons}, \eqref{eq:alphml}
\\
& {\mathcal{S}}_{\ul,\ell}^{\vFD} (\qb, {\boldsymbol{\theta}}, \boldsymbol \varsigma, \ALPHA) \geq q_{\ul,\ell}, \forall \ell  
\label{UL:QoS:cons:1}
\\
& q_{\ul,\ell} \geq \mathcal{S}_\ul^o, \forall \ell
\label{UL:QoS:cons:2}
\\
&\mathcal{S}_{\dl,k}^{\vFD} (\qa, \boldsymbol \theta, {\boldsymbol{\varsigma}}) \geq q_{\dl,k}, \forall k 
\label{DL:QoS:cons:1} 
\\
& q_{\dl,k} \geq \mathcal{S}_\dl^o, \forall k,
\label{DL:QoS:cons:2} 
\end{align}
\end{subequations}
where $\q_{\ul} \triangleq \{q_{\ul,\ell}\}, \q_{\dl} \triangleq \{q_{\dl,k}\}$ are auxiliary variables. The problem~\eqref{P:SE:equi} is a mixed-integer nonconvex optimization problem due to the binary variables involved. Moreover, there is a strong coupling between the continuous variables ($\boldsymbol \theta, \boldsymbol \varsigma, \boldsymbol \alpha$) and binary variables ($\qa, \qb$), which makes problem~\eqref{P:SE:equi} even more complicated. In what follows, we first transform problem \eqref{P:SE:equi} into a more tractable form by exploiting the special relationship between continuous and binary variables, and use successive convex approximation techniques to solve the transformed problem efficiently.

\subsection{Solution}\label{sec:SEsolution}
Let us introduce the additional nonnegative variables ${\OOmega} \triangleq \{{\omega}_{mq}\}, \bar{\OOmega} \triangleq \{\bar{\omega}_{mq}\}, \tilde{\OOmega} \triangleq \{\tilde{\omega}_{m\ell}\}, \hat{\OOmega} \triangleq \{\hat{\omega}_{m\ell k}\}, \tilde{\ALPHA} \triangleq \{\tilde{\alpha}_{mk}\}, \hat{\ALPHA} \triangleq \{\hat{\alpha}_{mk}\}, \bar{\ETA} \triangleq \{\bar{\eta}_{j\ell}\}, \hat{\ETA} \triangleq \{\hat{\eta}_{mjk\ell}\}$, where
\begin{align}
    \label{omega}
    & \omega_{m\ell}^2 \leq  b_m \varsigma_{\ell}, \forall m,\ell
    \\
    \label{omegabar}
        & \bar{\omega}_{mq}^2 \geq  b_m \varsigma_{q}, \forall m,q
    \\
    \label{omegatilde}
    & \omega_{m\ell} \alpha_{m\ell} \geq \tilde{\omega}_{m\ell}, \forall m,\ell
    \\
    \label{omegahat}
    & \bar{\omega}_{mq} \alpha_{m\ell} \leq \hat{\omega}_{m\ell q}, \forall m,\ell,q
    \\
    \label{alphatilde}
    & \alpha_{m\ell}^2 \leq \tilde{\alpha}_{m\ell}, \forall m,\ell
    \\
    \label{alphahat}
    & b_m \tilde{\alpha}_{m\ell} \leq \hat{\alpha}_{m\ell}, \forall m,\ell
    \\
    \label{etatilde}
    & \theta_{ik}^2 \leq \bar{\eta}_{ik}, \forall i,k
    \\
    \label{etahat}
    & \hat{\alpha}_{m\ell} \bar{\eta}_{ik} \leq \hat{\eta}_{mi\ell k}, \forall m,i,\ell,k,
\end{align}
which imply that
\begin{align}
    \label{omegatilde:2}
    & \sqrt{b_m \varsigma_{\ell}} \alpha_{m\ell} \geq \tilde\omega_{m\ell}, \forall m,\ell
    \\
    \label{omegahat:2}
    & b_m \varsigma_{q} \alpha_{m\ell}^2 \leq \hat{\omega}_{m\ell q}^2, \forall m,\ell,q
    \\
    \label{alphahat:2}
    & b_m \alpha_{m\ell}^2 \leq \hat{\alpha}_{m\ell}, \forall m,\ell
    \\
    \label{etahat:2}
    & b_m \alpha_{m\ell}^2 \theta_{ik}^2 \leq \hat{\eta}_{mi\ell k}, \forall m,i,\ell,k.
\end{align}
From \eqref{eq:UL:SE}, \eqref{omegatilde:2}--\eqref{etahat:2}, we have 
\begin{align}
\nonumber
    \mathcal{S}_{\ul,\ell}^{\vFD} (\qb, \boldsymbol \varsigma, \boldsymbol \theta, \boldsymbol \alpha)  &\geq \widetilde{\mathcal{S}}_{\ul,\ell}^{\vFD} (\tilde{\OOmega}, \hat{\OOmega}, \hat{\ETA}, \hat{\ALPHA})\nonumber\\
    &\triangleq  \frac{\tau_c-\tau_t}{\tau_c}\log_2
         \bigg(1+\frac{\Psi_{\ell}^2 (\tilde{\OOmega}) }{\Phi_{\ell}(\hat{\boldsymbol \omega}, \hat{\boldsymbol{\eta}}, \hat{\boldsymbol{\alpha}})}\bigg), \forall \ell,
\end{align}
where 
\begin{align}
\nonumber
    &\Psi_{\ell} (\tilde{\OOmega}) \triangleq \sqrt{\Nrx \rho_{u}} \sum_{\substack{m\in\mathcal{M}}} \tilde{\omega}_{m\ell} \gamuml,
    \\
\nonumber
    &\Phi_{\ell}(\hat{\boldsymbol \omega}, \hat{\boldsymbol{\eta}}, \hat{\boldsymbol{\alpha}}) \triangleq \rho_{u}
               \! \sum_{\substack{m\in\mathcal{M}}}
               \! \sum_{q\in\mathcal{K}_u}
                        \hat{\omega}_{m\ell q}^2
                        \betamlu
                        \gamma_{m\ell}^{\ul}
                        \nonumber\\
    &\hspace{7em}
    +
    \rho_{d}\Ntx
            \!\sum_{\substack{m\in\mathcal{M}}}
            \!\sum_{\substack{i\in\mathcal{M}}}
            \!\sum_{k\in\mathcal{K}_d}
                \hat{\eta}_{mi\ell k} \beta_{mi}\gamuml\gamma_{ik}^{\dl}
                \nonumber\\
    &\hspace{7em}
    +
     \sum_{\substack{m\in\mathcal{M}}}
                            \hat{\alpha}_{m\ell} \gamuml.
\end{align}   
Then, constraint \eqref{UL:QoS:cons:1} can be replaced by 
\vspace{0.25em}
\begin{align}
    \label{UL:QoS:cons:1:equi}
    \widetilde{\mathcal{S}}_{\ul,\ell}^{\vFD} (\tilde{\OOmega}, \hat{\OOmega}, \hat{\ETA}, \hat{\ALPHA}) \geq q_{\ul,\ell}, \forall \ell. 
\end{align}
By invoking \eqref{DL:power:cons}, we replace constraint \eqref{etaa:relation} by
\begin{align}
\label{etaa:relation:2}
N \gamdmk \theta_{mk}^2 \leq a_{m}, \quad\forall m,k.
\end{align}
To handle the binary constraints \eqref{a} and \eqref{b}, we observe that for any real number $x$, we have $x\in\{0,1\}\Leftrightarrow x-x^2=0\Leftrightarrow (x\in[0,1]\,~\&\,~x-x^2\leq0)$ \cite{vu18TCOM}. Thus, \eqref{a} and \eqref{b} can be replaced by the following equivalent constraint:
\begin{align}
\label{C}
& C(\qa, \qb) \triangleq \sum_{m\in\mathcal{M}} (a_{m}\!-\!a_{m}^2) + \sum_{m\in \mathcal{M}} (b_{m}\!-\!b_{m}^2) \leq 0
\\
\label{abrelax}
& 0 \leq a_{m} \leq 1,~0 \leq b_{m} \leq 1, \forall m.
\end{align}

Now, problem \eqref{P:SE:equi} can be written in a more tractable form as
\begin{align}\label{P:SE:equi:3}
\underset{\widetilde{\qx}\in \mathcal{F}}{\min}\,\, &
-\sum_{\ell\in\mathcal{K}_u} q_{\ul,\ell}   -  \sum_{k\in\mathcal{K}_d} q_{\dl,k},
\end{align}
where $\widetilde{\x} \triangleq \{\x, \q_{\ul}, \q_{\dl}, \OOmega, \bar{\OOmega}, \tilde{\OOmega}, \hat{\OOmega}, \tilde{\ALPHA}, \hat{\ALPHA}, \bar{\ETA}, \hat{\ETA}\}$, $\mathcal{F}\!\triangleq \!\{\eqref{sumab}, \eqref{theta}, \eqref{DL:power:cons},  \eqref{UL:power:cons}, \eqref{eq:alphml}, \eqref{UL:QoS:cons:2} - \eqref{DL:QoS:cons:2}, \eqref{omega} - \eqref{etahat}, \eqref{UL:QoS:cons:1:equi}-\eqref{abrelax}\}$ 
is a feasible set. To this end, we consider the following problem
\begin{align}\label{P:SE:equi:relax}
\underset{\widetilde{\qx}\in \widetilde{\mathcal{F}}}{\min}\,\, &
\mathcal{L}_{\mathtt{SE}}(\widetilde{\qx}),
\end{align}
where $\mathcal{L}_{\mathtt{SE}}(\widetilde{\qx})\triangleq -\sum_{\ell\in\mathcal{K}_u} q_{\ul,\ell}   -  \sum_{k\in\mathcal{K}_d} q_{\dl,k} + \lambda C(\qa, \qb)$ is the Lagrangian of \eqref{P:SE:equi:3} and $\lambda$ is the Lagrangian multiplier corresponding to constraint \eqref{C}. Here, $\widetilde{\mathcal{F}}\triangleq \mathcal{F}\setminus \{\eqref{C}\}$.

\begin{proposition}
\label{proposition-dual}
The values $C_{\lambda}$ of $C$ at the solution of \eqref{P:SE:equi:relax} corresponding to $\lambda$ converge to $0$ as $\lambda \rightarrow +\infty$. Also, problem \eqref{P:SE:equi:3} has strong duality, i.e.,
\begin{equation}\label{Strong:Dualitly:hold}
\underset{\widetilde{\qx}\in\mathcal{F}}{\min}\,\,
-\sum_{\ell\in\mathcal{K}_u} q_{\ul,\ell}   -  \sum_{k\in\mathcal{K}_d} q_{\dl,k}
=
\underset{\lambda\geq0}{\sup}\,\,
\underset{\widetilde{\qx}\in\widetilde{\mathcal{F}}}{\min}\,\,
\mathcal{L}_{\mathtt{SE}}(\widetilde{\qx}).
\end{equation}
Then, problem \eqref{P:SE:equi:3} is equivalent to problem \eqref{P:SE:equi:relax} at the optimal solution $\lambda^* \geq0$ of the sup-min problem in \eqref{Strong:Dualitly:hold}.
\end{proposition}

\begin{proof}
    The proof has a similar procedure as the proof of \cite[Proposition 1]{vu18TCOM}, and hence, omitted.
\end{proof}
Note that it is theoretically required to have $C_{\lambda}=0$ in order to obtain the optimal solution to problem \eqref{P:SE:equi:3}. According to Proposition~\ref{proposition-dual}, $C_{\lambda}$ converges to $0$ as $\lambda\to+\infty$. For practical implementation, it is  acceptable for $C_{\lambda}$ to be sufficiently small with a sufficiently large value of $\lambda$. In our numerical experiments, for $\varepsilon = 5 \times 10^{-5}$, we see that $\lambda=1$ is enough to ensure that $C_{\lambda}/(MK) \leq \varepsilon$. This way of selecting $\lambda$ has been widely used in the literature, e.g., see \cite{vu18TCOM} and references therein.

Problem \eqref{P:SE:equi:relax} is still difficult to solve due to the non-convex constraints \eqref{DL:QoS:cons:1} and \eqref{UL:QoS:cons:1:equi}. To deal with constraint \eqref{DL:QoS:cons:1}, we observe that  
\begin{align}
\log \bigg(1 + \frac{x^2}{y}\bigg)  &\geq  \log
             \bigg(1+\frac{(x^{(n)})^2}{y^{(n)}}\bigg) -
            \frac{(x^{(n)})^2}{y^{(n)}}
            \nonumber\\
&\hspace{0em}
           		 + 2\frac{x^{(n)}x}{y^{(n)}}
           		 - \frac{(x^{(n)})^2(x^2
           		 + y)}{y^{(n)}((x^{(n)})^2
           		 +y^{(n)})}, 
\end{align}
where $x > 0, y > 0$ \cite[Eq. (40)]{vu20TWC}.
Therefore, $\mathcal{S}_{\dl,k}^{\vFD}  (\qa, \boldsymbol \theta, {\boldsymbol{\varsigma}})$ has a concave lower bound $\widehat{\mathcal{S}}_{\dl,k}^{\vFD}  (\qa, \boldsymbol \theta, {\boldsymbol{\varsigma}})$ that is given as
\begin{align}
    \widehat{\mathcal{S}}_{\dl,k}^{\vFD}  (\boldsymbol \theta, {\boldsymbol{\varsigma}})
     &\triangleq
     \frac{\tau_c - \tau_{t}}{\tau_c\log 2} 
     \Bigg[
      \log\bigg(1\!+\!
                  \frac{(\Xi_k^{(n)})^2}{\Omega_k^{(n)}}\bigg)
                  \!-\!
                  \frac{(\Xi_k^{(n)})^2}{\Omega_k^{(n)}}
                  \nonumber\\
&\hspace{0em}
                  \!+\!
                  2\frac{\Xi_k^{(n)}\Xi_k}{\Omega_k^{(n)}}
                   \!-\!
                   \frac{(\Xi_k^{(n)})^2(\Xi_k^2 + \Omega_k)}{\Omega_k^{(n)}((\Xi_k^{(n)})^2
                   \!+\!
                   \Omega_k^{(n)})}
       \Bigg],
\end{align}
where $\Xi_k$ and $\Omega_k$ are defined in \eqref{eq:DL:SE}.
Then, constraint \eqref{DL:QoS:cons:1}  is approximated by the following convex constraint
\begin{align}
    \label{DL:QoS:cons:approx}
    \widehat{\mathcal{S}}_{\dl,k}^{\vFD} (\boldsymbol \theta,  {\boldsymbol{\varsigma}}) \geq q_{\dl,k}, \forall k.
\end{align}
Similarly, to deal with constraint \eqref{UL:QoS:cons:1:equi}, we see that the concave lower bound of $\widetilde{\mathcal{S}}_{\ul,\ell} ^{\vFD}  (\tilde{\OOmega}, \hat{\OOmega}, \hat{\ETA}, \hat{\ALPHA})$ is given by 
\begin{align}
    \widehat{\mathcal{S}}_{\ul,\ell} ^{\vFD} 
    (\tilde{\OOmega}, \hat{\OOmega}, \hat{\ETA}, \hat{\ALPHA})
    &\triangleq
    \frac{\tau_c \!-\! \tau_{t}}{\tau_c\log 2}
    \Bigg[
          \log
             \bigg(1 \!+\! \frac{(\Psi_{\ell}^{(n)})^2}{\Phi_{\ell}^{(n)}}\bigg) 
             \nonumber\\
&\hspace{-7em}
-
            \frac{(\Psi_{\ell}^{(n)})^2}{\Phi_{\ell}^{(n)}}
           		 + 2\frac{\Psi_{\ell}^{(n)}\Psi_{\ell}}{\Phi_{\ell}^{(n)}}
           		 - \frac{(\Psi_{\ell}^{(n)})^2(\Psi_{\ell}^2
           		 + \Phi_{\ell})}{\Phi_{\ell}^{(n)}((\Psi_{\ell}^{(n)})^2
           		 +\Phi_{\ell}^{(n)})}
     \Bigg].
\end{align}
Then, constraint \eqref{UL:QoS:cons:1:equi} is then approximated by the following convex constraint
\begin{align}
    \label{UL:QoS:cons:approx}
    \widehat{\mathcal{S}}_{\ul,\ell}^{\vFD} (\tilde{\OOmega}, \hat{\OOmega}, \hat{\ETA}, \hat{\ALPHA}) \geq q_{\ul,\ell}, \forall \ell.
\end{align}

By invoking the following lower bounds~\cite{vu20TWC} 
\begin{align}
    \label{xy:ub}
    & xy \leq 0.25 [(x+y)^2-2(x^{(n)}-y^{(n)})(x-y) 
    \nonumber\\
&\hspace{2em}+ (x^{(n)}-y^{(n)})^2]
    \\
    \label{minusxy:ub}
    & -xy \leq 0.25 [(x-y)^2-2(x^{(n)}\nonumber\\
&\hspace{2em}+\!y^{(n)})(x+y)
    + (x^{(n)}+y^{(n)})^2],
\end{align}
where $\forall x\geq0, y\geq0$, the convex upper bound of $C(\qa, \qb)$ is given by
\begin{align}
    &\widetilde{C}(\qa, \qb) \triangleq \sum_{m\in\mathcal{M}} \left[a_{m}-2a_{m}^{(n)}a_{m} + (a_{m}^{(n)})^2\right] \nonumber\\
&\hspace{5em}+ \sum_{m\in \mathcal{M}}\left[b_{m}-2b_{m}^{(n)}b_{m} + (b_{m}^{(n)})^2\right].
\end{align}
Similarly, constraints \eqref{omega}--\eqref{omegahat}, \eqref{alphahat}, and \eqref{etahat} can be approximated by the following convex constraints
\begin{align}
    \label{omega:1:convex}
    & \omega_{m\ell}^2 + 0.25 [(b_m-\varsigma_{\ell})^2-2(b_m^{(n)}
        +\varsigma_{\ell}^{(n)})(b_m+\varsigma_{\ell})
\nonumber\\
    &\hspace{3em}
    + (b_m^{(n)}+\varsigma_{\ell}^{(n)})^2]
    \leq  0, \forall m,\ell
    \\
    & \!\!\!\! 0.25 [(b_m \!+\! \varsigma_{q})^2 -2(b_m^{(n)}-\varsigma_{q}^{(n)})(b_m-\varsigma_{q}) + (b_m^{(n)}-\varsigma_{q}^{(n)})^2] 
    \nonumber\\
    &\hspace{3em}
    - 2\bar{\omega}_{m{q}}^{(n)}\bar{\omega}_{m{q}} + (\bar{\omega}_{m{q}}^{(n)})^2\leq 0, \forall m,q
    \\
    \label{omegatilde:convex}
    & \!\!\!\! \tilde{\omega}_{m\ell} + 0.25 [(\omega_{m\ell}-\alpha_{m\ell})^2 - 2(\omega_{m\ell}^{(n)}+\alpha_{m\ell}^{(n)})(\omega_{m\ell}+\alpha_{m\ell})
    \nonumber\\
    &\hspace{3em}
    + (\omega_{m\ell}^{(n)}+\alpha_{m\ell}^{(n)})^2] \leq 0, \forall m,\ell
    \\
    \label{omegahat:convex}
    & \!\!\!\! 0.25 [(\bar{\omega}_{mq}+\alpha_{m\ell})^2-2(\bar{\omega}_{mq}^{(n)}-\alpha_{m\ell}^{(n)})(\bar{\omega}_{mq}-\alpha_{m\ell}) 
    \nonumber\\
    &\hspace{3em}
    + (\bar{\omega}_{mq}^{(n)}-\alpha_{m\ell}^{(n)})^2]
     - \hat{\omega}_{m\ell q} \leq 0, \forall m,\ell,q
    \\
    \label{alphahat:convex}
    & \!\!\!\! 0.25 [(b_m+\tilde{\alpha}_{m\ell})^2-2(b_m^{(n)}-\tilde{\alpha}_{m\ell}^{(n)})(b_m-\tilde{\alpha}_{m\ell}) 
    \nonumber\\
    &\hspace{3em}
    + (b_m^{(n)}-\tilde{\alpha}_{m\ell}^{(n)})^2]
    - \hat{\alpha}_{m\ell} \leq 0, \forall m,\ell
    \\
    \label{etahat:convex}
    & \!\!\!\! 0.25 [(\hat{\alpha}_{m\ell}+\bar{\eta}_{ik})^2-2(\hat{\alpha}_{m\ell}^{(n)}-\bar\eta_{ik}^{(n)})(\hat{\alpha}_{mk}-\bar\eta_{ik}) 
    \nonumber\\
    &\hspace{3em}
    + (\hat{\alpha}_{m\ell}^{(n)}-\bar\eta_{ik}^{(n)})^2]
     - \hat{\eta}_{mi\ell k} \leq 0, \forall m,i,\ell,k.
\end{align}

At iteration $(n+1)$, for a given point $\widetilde{\qx}^{(n)}$, problem \eqref{P:SE:equi:relax} can finally be approximated by the following convex problem:
\begin{align}
\label{P:SE:equi:relax:approx}
\underset{\widetilde{\qx}\in\widehat{\mathcal{F}}}{\min} \,\,
& \widehat{\mathcal{L}}_{\mathtt{SE}^{\vFD}}(\widetilde{\qx}),
\end{align}
where $\widehat{\mathcal{L}}_{\mathtt{SE}^{\vFD}}(\widetilde{\qx}) =- \sum_{\ell\in\mathcal{K}_u} q_{\ul,\ell}   -  \sum_{k\in\mathcal{K}_d} q_{\dl,k} + \lambda\widetilde{C}(\qa, \qb)$, $\widehat{\mathcal{F}} \triangleq \{\widetilde{\mathcal{F}}, \eqref{DL:QoS:cons:approx}, \eqref{UL:QoS:cons:approx},  \eqref{omega:1:convex}-\eqref{etahat:convex}\} \setminus \{\eqref{DL:QoS:cons:1}, \eqref{omega}\\-\eqref{omegahat}, \eqref{alphahat}, \eqref{etahat}, \eqref{UL:QoS:cons:1:equi}\}$ is a convex feasible set. In Algorithm~\ref{alg}, we outline the main steps to solve problem \eqref{P:SE:equi:3}.
Starting from a random point $\widetilde{\qx}\in\widehat{\mathcal{F}}$, we solve \eqref{P:SE:equi:relax:approx} to obtain its optimal solution $\widetilde{\qx}^*$, and use $\widetilde{\qx}^*$ as an initial point in the next iteration. The algorithm terminates when an accuracy level is reached. Algorithm~\ref{alg} will converge to a stationary point, i.e., a Fritz John solution, of problem \eqref{P:SE:equi:relax} (hence \eqref{P:SE:equi:3} or \eqref{P:SE}). The proof of this convergence property uses similar steps in the proof of \cite[Proposition 2]{vu18TCOM}, and hence, is omitted due to lack of space. 

Algorithm~\ref{alg} requires solving a series of convex problems \eqref{P:SE:equi:relax:approx}. For ease of presentation, if we let $K_d = K_u = K$, problem \eqref{P:SE:equi:relax:approx}  can be transformed to an equivalent problem that involves $A_v\triangleq 2M + 3K + 9MK + MK^2 + M^2K^2$ real-valued scalar variables, $A_l\triangleq 5M + 4K + 9MK + MK^2 + M^2K^2$ linear constraints, $A_q\triangleq M + 2K + 7MK + MK^2 + M^2K^2$ quadratic constraints. Therefore, the algorithm for solving problem \eqref{P:SE:equi:relax} requires a complexity of $\OO(\sqrt{A_l+A_q}(A_v+A_l+A_q)A_v^2)$.

\begin{algorithm}[!t]
\caption{Solving problem \eqref{P:SE:equi:relax}}
\begin{algorithmic}[1]
\label{alg}
\STATE \textbf{Initialize}: $n\!=\!0$, 
$\lambda > 1$, a random point $\widetilde{\qx}^{(0)}\!\in\!\widehat{\mathcal{F}}$.
\REPEAT
\STATE Update $n=n+1$
\STATE Solve \eqref{P:SE:equi:relax:approx} to obtain its optimal solution $\widetilde{\qx}^*$
\STATE Update $\widetilde{\qx}^{(n)}=\widetilde{\qx}^*$
\UNTIL{convergence}
\end{algorithmic}
\vspace{0.5em}
\end{algorithm}
\begin{remark}[Initial point and infeasible SE maximization problem] \label{remark:SEfeasible}
Our Algorithm~\ref{alg} needs a feasible point $\widetilde{\x} \in \widehat{\FF}$ to start. It is easy to find a point $\widetilde{\x} \in \widehat{\FF} \setminus \{\eqref{UL:QoS:cons:2}, \eqref{DL:QoS:cons:2}\}$ by using a random procedure and letting the constraints in $\widehat{\FF} \setminus \{\eqref{UL:QoS:cons:2}, \eqref{DL:QoS:cons:2}\}$ happen with equality. However, when the minimum individual SEs required for QoS, i.e., $\mathcal{S}_\ul^o$ and $\mathcal{S}_\dl^o$, are large but the UEs have unfavourable links to the APs, the QoS constraints \eqref{UL:QoS:cons:2} and \eqref{DL:QoS:cons:2} are not easy to satisfy. In this case, we start with a random point $\widetilde{\x} \in \widehat{\FF} \setminus \{\eqref{UL:QoS:cons:2}, \eqref{DL:QoS:cons:2}\}$ and solve the following problem (instead of problem \eqref{P:SE:equi:relax:approx}) in each iteration of Algorithm~\ref{alg}
\begin{subequations}\label{P:SE:equi:relax:approx:s}
\begin{align}
\underset{\widetilde{\qx}\in\widehat{\mathcal{F}}\setminus \{\eqref{UL:QoS:cons:2}, \eqref{DL:QoS:cons:2}\}, \z_{\ul}, \z_{\dl}}{\min} \,\,
& \widehat{\mathcal{L}}_{\mathtt{SE}^{\vFD}}(\widetilde{\qx}) + \phi z,
\\ 
\mathrm{s.t.} \,\,\,\,\,\,\,\,\,\,\,\,\,\,\,\,\,\,\,\,
\label{qul:s}
& q_{\ul,\ell} + z_{\ul,\ell} \geq \mathcal{S}_\ul^o, \forall \ell
\\
\label{qdl:s}
& q_{\dl,k} + z_{\dl,k} \geq \mathcal{S}_\dl^o, \forall k
\\
\label{s:cons}
& z_{\ul,\ell} \geq 0, z_{\dl,k} \geq 0, \forall k,\ell,
\end{align}
\end{subequations}
where $z \triangleq \sum_{\ell\in\K_u} z_{\ul,\ell} + \sum_{k\in\K_d} z_{\dl,k}$ and $\phi > 0$ is a penalty parameter. Here, $\z_{\ul} \triangleq \{z_{\ul,\ell}\}, \z_{\dl} \triangleq \{z_{\dl,k}\}$ are additional variables that makes constraints \eqref{UL:QoS:cons:2}, \eqref{DL:QoS:cons:2} satisfied if they are sufficiently small. Since $\z_{\ul},\z_{\dl}$ are nonnegative and \eqref{P:SE:equi:relax:approx:s} is a minimization problem, $\z_{\ul},\z_{\dl}$ are forced to approach $0$ during the iterative process of Algorithm~\ref{alg}. When Algorithm~\ref{alg} converges, if  $z$ is smaller than a predefined error threshold, the problem \eqref{P:SE:equi:3} or \eqref{P:SE} is feasible with constraints \eqref{UL:QoS:cons:2}, \eqref{DL:QoS:cons:2} satisfied, and we take the converged point as the final solution. Otherwise, problem \eqref{P:SE:equi:3} or \eqref{P:SE} is considered as an infeasible problem. 
\end{remark}

\section{Energy Efficiency Maximization}
\label{sec:EE}
\subsection{Problem Formulation}
In this subsection, we aim at optimizing the mode assignment of the APs $(\aaa,\bb)$, power coefficients $(\THeta, \VARSIGMA)$, and LSFD weight $\ALPHA$ to maximize the total EE, under the constraints on QoS requirements for each UE, maximum transmit power at each DL AP and each UL UE. The total EE (in bit/Joule) is defined as the sum throughput (bit/s) divided by the total power consumption (Watt) in the network
\begin{align}~\label{eq:EE:def}
   \mathtt{EE}^{\vFD}(\x) =\frac{B.\mathcal{S}^{\vFD}(\qx)}{\frac{\tau_c-\tau_t}{\tau_c}P_\mathtt{total}^{\vFD} (\x) },
\end{align}
where $\qx$ is defined in~\eqref{eq:Sx}. More precisely, the optimization problem is formulated
as follows:
\begin{subequations}
\begin{align}\label{P:EE}
\underset{\qx}{\max}\,\, &\hspace{1em}
	\mathtt{EE}^{\vFD} (\x)
	  \\
	\mathrm{s.t.} \,\,
	&\hspace{1em} \eqref{a}-\eqref{sumab}, \eqref{theta},
	\eqref{etaa:relation}, \eqref{DL:power:cons}, \eqref{UL:power:cons}, \eqref{eq:alphml}, \eqref{UL:QoS:cons}, \eqref{DL:QoS:cons}.
\end{align}    
\end{subequations}

Problem~\eqref{P:EE} is also a nonconvex mixed-integer problem. However, it has a tight coupling of the AP mode assignment variables ($\qa, \qb$) and the power consumption of backhaul signalling loads, which is not the case in the SE maximization problem~\eqref{P:SE}. On one hand, this is the issue that makes the mathematical structure of problem~\eqref{P:EE} significantly different from that of problem~\eqref{P:SE}. Thus, we cannot apply straightforwardly the proposed Algorithm~\ref{alg} to solve problem~\eqref{P:EE}. On the other hand, this issue makes problem~\eqref{P:EE} technically much more challenging than problem~\eqref{P:SE} and difficult to find its optimal solution. Therefore, instead of finding the optimal solution to the EE problem, we aim to find its suboptimal solution. 

First, we see that by the definition of $\PbhvFD$ in \eqref{PbhvFD}, we have
\begin{align}
    \label{Pbhfull}
    \PbhvFD < \PbhfulvFD \triangleq B\sum_{m\in\MM} \SSS^{\vFD} (\x) P_{\mathtt{bt},m},
\end{align}
which is the rate-dependent power consumption when each AP shares with the CPU full backhaul signaling loads for all the DL and UL UEs. Then, $\PtotvFD(\x)$ is always smaller than 
\begin{align}
    \label{eq:Ptotal:upper}
    \PtotbhvFD (\x)
    & \triangleq 
      \sum_{m\in\mathcal{M}} \frac{N\rho_{d}\Sn}{\zeta_m}\left(\sum_{k\in\mathcal{K}_{d}} \gamdmk \theta_{mk}^2 
    \right) 
    \nonumber\\
    &\hspace{2em}
    +\sum_{\ell\in\mathcal{K}_{u}} \frac{\rho_u\Sn}{\chi} {\varsigma}_{\ell} +\PUfix 
    + \PbhfulvFD
    \nonumber\\
    & \hspace{2em} 
    +  \sum_{m\in\mathcal{M}} a_m (N P_{\mathtt{cdl},m} + P_{\mathtt{fdl},m}) 
    \nonumber\\
    &\hspace{2em}
    + \sum_{m\in\mathcal{M}}  b_m (NP_{\mathtt{cul},m} + P_{\mathtt{ful},m}).
\end{align}
Therefore, we have
\begin{align}
    \EEEE_{fullbh}^{\vFD} (\x) \triangleq \frac{B.\mathcal{S}^{\vFD}(\qx)}{\frac{\tau_c-\tau_t}{\tau_c}\PtotbhvFD (\x) } < \EEEE^{\vFD}(\x).
\end{align}
Now, instead of solving problem~\eqref{P:EE}, we aim to solve the following problem
\begin{subequations}
\begin{align}\label{P:EE:fullbh}
\underset{\qx}{\max}\,\, &\hspace{1em}
	\mathtt{EE}_{fullbh}^{\vFD} (\x)
	  \\
	\mathrm{s.t.} \,\,
	&\hspace{1em} \eqref{a}-\eqref{sumab}, \eqref{theta},
	\eqref{etaa:relation}, \eqref{DL:power:cons}, \eqref{UL:power:cons}, \eqref{eq:alphml}, \eqref{UL:QoS:cons}, \eqref{DL:QoS:cons}.
\end{align}    
\end{subequations}
Note that the solution to problem~\eqref{P:EE:fullbh} is not the optimal solution to problem~\eqref{P:EE} but can be sufficiently close to this solution, which is shown later in the simulation results of Section~\ref{Sec:Numer}. 
For problem~\eqref{P:EE:fullbh}, we observe that 
\begin{align}
    \label{EE:1}
    \mathtt{EE}_{fullbh}^{\vFD}&=\frac{B.\mathcal{S}^{\vFD}(\qx)}{ \frac{\tau_c-\tau_t}{\tau_c} 
    \PtotbhvFD (\x)} =\frac{1}{\frac{\tau_c-\tau_t}{\tau_c} \frac{\PtotbhvFD(\x)}{B.\mathcal{S}^{\vFD}(\qx)}}
    \nonumber\\
&\hspace{0em}
    = \frac{1}{ \frac{\tau_c-\tau_t}{\tau_c} \left( \frac{\widetilde{P}(\x)}{B.\mathcal{S}^{\vFD}(\qx)} 
    + \sum_{m\in\mathcal{M}}P_{\mathtt{bt},m} \right)},
\end{align}
where 
\begin{align}
    \nonumber
    \widetilde{P}(\x) \triangleq
    &\sum_{m\in\mathcal{M}} \frac{N\rho_{d}\Sn}{\zeta_m}\left(\sum_{k\in\mathcal{K}_{d}} \gamdmk \theta_{mk}^2 
    \right) +\sum_{\ell\in\mathcal{K}_{u}} \frac{\rho_u\Sn}{\chi} {\varsigma}_{\ell} 
    \nonumber\\
   &\hspace{0em}+\PUfix 
    +  \sum_{m\in\mathcal{M}} a_m (N P_{\mathtt{cdl},m} + P_{\mathtt{fdl},m}) 
    \nonumber\\
&\hspace{0em}
  + \sum_{m\in\mathcal{M}}  b_m (NP_{\mathtt{cul},m} + P_{\mathtt{ful},m}).
\end{align}
By invoking~\eqref{EE:1}, problem~\eqref{P:EE:fullbh} can be rewritten as
\begin{subequations}
\begin{align}\label{P:EE:fullbh:1}
\underset{\qx}{\max}\,\, &\hspace{1em}
	\frac{B.\mathcal{S}^{\vFD}(\qx)}{\widetilde{P}(\x)}
	  \\
	\mathrm{s.t.} \,\,
	&\hspace{1em} \eqref{a}-\eqref{sumab}, \eqref{theta},
	\eqref{etaa:relation}, \eqref{DL:power:cons}, \eqref{UL:power:cons}, \eqref{eq:alphml}, \eqref{UL:QoS:cons}, \eqref{DL:QoS:cons}.
\end{align}    
\end{subequations}
Problem~\eqref{P:EE:fullbh:1} is then equivalent to 
\begin{subequations}
\label{P:EE:fullbh:2}
	\begin{align}
	\underset{\x, t, \hat{p}}{\min}\,\, &
		- t
		  \\
		\mathrm{s.t.} \,\,
		&\hspace{1em} \eqref{a}-\eqref{sumab}, \eqref{theta},
		\eqref{etaa:relation}, \eqref{DL:power:cons}, \eqref{UL:power:cons}, \eqref{eq:alphml}, \eqref{UL:QoS:cons:1}-\eqref{DL:QoS:cons:2}
		\\
		\label{tphat}
		&\hspace{1em} t \hat{p} \leq  \sum_{\ell\in\mathcal{K}_u} B q_{\ul,\ell}   +  \sum_{k\in\mathcal{K}_d} B q_{\dl,k} 
		\\
		\label{phat}
		&\hspace{1em} \hat{p} \geq \widetilde{P} (\x),
	\end{align}
\end{subequations}
where $t$ and $\hat{p}$ are additional variables. In the following, we transform problem \eqref{P:EE:fullbh:2} into a more tractable form which is then solved by successive convex approximation techniques. 

\subsection{Solution}
Using similar steps to transform the problem of maximizing the sum SE into a more tractable one as discussed in Section~\ref{sec:SEsolution}, problem \eqref{P:EE:fullbh:2} can be rewritten as 
\begin{align}\label{P:EE:fullbh:3}
\underset{\widehat{\x}\in \mathcal{H}}{\min}\,\, &
-t,
\end{align}
where $\mathcal{H} \! \triangleq \! \{\eqref{a}-\eqref{sumab}, \eqref{theta},  \eqref{etaa:relation}, \eqref{DL:power:cons}, \eqref{UL:power:cons}, \eqref{eq:alphml}, \eqref{UL:QoS:cons:2} - \eqref{DL:QoS:cons:2}, \eqref{omega} - \eqref{etahat}, \eqref{UL:QoS:cons:1:equi}-\eqref{abrelax},\eqref{tphat}, \eqref{phat}\}$.
Now, we consider the following problem
\begin{align}\label{P:EE:fullbh:3:relax}
\underset{\widehat{\x}\in \HHH}{\min}\,\, &
\mathcal{L}_{\mathtt{EE}^{\vFD}}(\widehat{\x}),
\end{align}
where $\mathcal{L}_{\mathtt{EE}^{\vFD}}(\widehat{\x})\triangleq -t + \lambda C(\qa, \qb)$ is the Lagrangian of \eqref{P:EE:fullbh:3} and $\lambda$ is the Lagrangian multiplier corresponding to constraint \eqref{C}. Here, $\widetilde{\HHH}\triangleq \HHH\setminus \{\eqref{C}\}$.

\begin{proposition}
\label{proposition-dual2}
The values $C_{\lambda}$ of $C$ at the solution of \eqref{P:EE:fullbh:3:relax} corresponding to $\lambda$ converge to $0$ as $\lambda \rightarrow +\infty$. Also, problem \eqref{P:EE:fullbh:3} has strong duality, i.e.,
\begin{equation}\label{Strong:Dualitly:hold2}
\underset{\widehat{\qx}\in\HHH}{\min}\,\,
-t
=
\underset{\lambda\geq0}{\sup}\,\,
\underset{\widehat{\qx}\in\widetilde{\HHH}}{\min}\,\,
\mathcal{L}_{\mathtt{EE}^{\vFD}}(\widehat{\qx}).
\end{equation}
Then, problem \eqref{P:EE:fullbh:3} is equivalent to problem \eqref{P:EE:fullbh:3:relax} at the optimal solution $\lambda^* \geq0$ of the sup-min problem in \eqref{Strong:Dualitly:hold2}.
\end{proposition}
\noindent
Similarly, the proof of Proposition~\ref{proposition-dual2} follows \cite{vu18TCOM}, and hence, omitted. According to Proposition~\ref{proposition-dual2}, $C_{\lambda}$ converges to $0$ as $\lambda\to+\infty$, and the optimal solution to problem \eqref{P:EE:fullbh:3} is obtained. We recall that for practical implementation, it is acceptable for $C_{\lambda}$ to be sufficiently small with a sufficiently large value of $\lambda$. In our numerical experiments, for $\varepsilon = 5\times 10^{-5}$, we see that $\lambda=10$ is enough to ensure that $C_{\lambda}/(MK) \leq \varepsilon$. 

From \eqref{xy:ub}, the nonconvex constraint \eqref{tphat} can be approximated by the following convex constraint
\begin{align}
    \label{tphat:convex}
    &0.25 [(t+\hat{p})^2-2(t^{(n)}-\hat{p}^{(n)})(t-\hat{p}) + (t^{(n)}-\hat{p}^{(n)})^2]    \nonumber\\
&\hspace{2em}
\leq B  \sum_{\ell\in\mathcal{K}_u} q_{\ul,\ell}   + \sum_{k\in\mathcal{K}_d} q_{\dl,k}.
\end{align}
To deal with the other nonconvex constraints of problem \eqref{P:EE:fullbh:3:relax}, similar approximation techniques in Section~\ref{Sec:Numer} are used. Finally, at iteration $(n+1)$, for a given point $\widehat{\qx}^{(n)}$, problem \eqref{P:EE:fullbh:3:relax} can finally be approximated by the following convex problem:
\begin{align}
\label{P:EE:equi:relax:approx}
\underset{\widehat{\qx}\in\widehat{\mathcal{H}}}{\min} \,\,
& \widehat{\mathcal{L}}_{\mathtt{EE}^{\vFD}}(\widehat{\qx}),
\end{align}
where $\widehat{\mathcal{L}}_{\mathtt{EE}^{\vFD}}(\widetilde{\qx}) = - t + \lambda \widetilde{C}(\qa, \qb)$, $\widehat{\mathcal{H}} \triangleq \{\widehat{\FF}, \eqref{tphat:convex}\} \setminus \{\eqref{tphat}\}$ is a convex feasible set. In Algorithm~\ref{alg:EE}, we outline the main steps to solve problem \eqref{P:EE:fullbh:3}.
Starting from a random point $\widehat{\qx}\in\widehat{\mathcal{H}}$, we solve \eqref{P:EE:equi:relax:approx} to obtain its optimal solution $\widehat{\qx}^*$, and use $\widehat{\qx}^*$ as an initial point in the next iteration. 
The algorithm terminates when an accuracy level is reached. Algorithm~\ref{alg:EE} will converge to a stationary point, i.e., a Fritz John solution, of problem \eqref{P:EE:fullbh:3} (hence \eqref{P:EE:fullbh:2} or \eqref{P:EE:fullbh}). The proof of this convergence property also uses similar steps in the proof of \cite[Proposition 2]{vu18TCOM}, and hence, is omitted due to lack of space. 

\begin{algorithm}[!t]
\caption{Solving problem \eqref{P:EE:fullbh:3:relax}}
\begin{algorithmic}[1]
\label{alg:EE}
\STATE \textbf{Initialize}: $n\!=\!0$, 
$\lambda > 1$, a random point $\widehat{\qx}^{(0)}\!\in\!\widehat{\mathcal{H}}$.
\REPEAT
\STATE Update $n=n+1$
\STATE Solve \eqref{P:EE:equi:relax:approx} to obtain its optimal solution $\widehat{\qx}^*$
\STATE Update $\widehat{\qx}^{(n)}=\widehat{\qx}^*$
\UNTIL{convergence}
\end{algorithmic}
\end{algorithm}

Algorithm~\ref{alg:EE} requires solving a series of convex problems \eqref{P:EE:equi:relax:approx}. Problem \eqref{P:EE:equi:relax:approx}  can be transformed to an equivalent problem that involves $A_v\triangleq 2M + 3K + 9MK + MK^2 + M^2K^2 + 2$ real-valued scalar variables, $A_l\triangleq 5M + 4K + 9MK + MK^2 + M^2K^2 + 1$ linear constraints, $A_q\triangleq M + 2K + 7MK + MK^2 + M^2K^2 + 1$ quadratic constraints. Therefore, the algorithm for solving problem \eqref{P:SE:equi:relax}  requires a complexity of $\OO(\sqrt{A_l+A_q}(A_v+A_l+A_q)A_v^2)$ \cite{tam16TWC}.

\begin{remark}[Initial point and infeasible EE maximization problem] \label{remark:EEfeasible}
We recall that when the individual SE requirements $\mathcal{S}_\ul^o$ and $\mathcal{S}_\dl^o$ are large but the UEs have unfavourable links to the APs, the QoS constraints \eqref{UL:QoS:cons:2} and \eqref{DL:QoS:cons:2} are difficult to satisfy. In this case, we use a similar procedure as discussed in Remark~\ref{remark:SEfeasible} for Algorithm~\ref{alg:EE}. Specifically, we start with a random point $\widehat{\x} \in \widehat{\HHH} \setminus \{\eqref{UL:QoS:cons:2}, \eqref{DL:QoS:cons:2}\}$ and solve the following problem (instead of problem \eqref{P:EE:equi:relax:approx}) in each iteration of Algorithm~\ref{alg:EE}
\begin{subequations}
\begin{align}\label{P:EE:equi:relax:approx:s}
\underset{\widehat{\qx}\in\widehat{\HHH}\setminus \{\eqref{UL:QoS:cons:2}, \eqref{DL:QoS:cons:2}\}, \z_{\ul}, \z_{\dl}}{\min} \,\,
& \widehat{\mathcal{L}}_{\mathtt{EE}^{\vFD}}(\widehat{\qx}) + \phi z,
\\ 
\mathrm{s.t.} \,\,\,\,\,\,\,\,\,\,\,\,\,\,\,\,\,\,\,
& \eqref{qul:s}, \eqref{qdl:s}, \eqref{s:cons},
\end{align}    
\end{subequations}
where $z$ is defined in Remark~\ref{remark:SEfeasible}. When Algorithm~\ref{alg:EE} converges, if $z$ is smaller than a predefined error tolerance, the problem \eqref{P:EE:fullbh:2} or \eqref{P:EE:fullbh} is feasible and we take the converged point as the final solution. Otherwise, problem \eqref{P:EE:fullbh:2} or \eqref{P:EE:fullbh} is considered as an infeasible problem. 
\end{remark}

\section{Baseline Schemes}\label{sec:bench}
To investigate the effectiveness of our proposed optimized network-assisted full-duplex ($\vvFD$) scheme for CF-mMIMO systems, we introduce the following baseline schemes for comparisons in the numerical results of Section~\ref{Sec:Numer}. 

\subsection{Network-Assisted Full-Duplex CF-mMIMO Systems: Heuristic Approaches}
To show the advantages of the joint optimization of AP mode assignment, power control, and LSFD weights in our $\vvFD$ scheme, we consider two heuristic NAFD schemes as follows.

\subsubsection{Network-assisted full-duplex with random AP mode assignment ($\RvFD$)} 
In this scheme, we assume that the AP modes $(\aaa,\bb)$ are randomly assigned. Accordingly, we optimize the power control coefficients $(\ETA,\VARSIGMA)$ and LSFD weights $\ALPHA$, under the same SE requirement constraints for UL and DL UEs. The problems of sum SE maximization of the $\RvFD$ scheme for the given random mode assignment vectors $(\aaa,\bb)$ can be respectively expressed as
\begin{subequations}
\begin{align}\label{P:SE:RvFD}
\underset{\THeta, \VARSIGMA, \ALPHA, \qq_{\ul},\qq_{\dl}}{\min}\,\, &
- \sum_{\ell\in\mathcal{K}_u} q_{\ul,\ell}   -  \sum_{k\in\mathcal{K}_d} q_{\dl,k}  \\
\mathrm{s.t.} \,\,\,\,\,\,\,\,\,
& \eqref{a}-\eqref{sumab}, \eqref{theta},
		\eqref{etaa:relation}, \eqref{DL:power:cons}, \eqref{UL:power:cons}, \eqref{eq:alphml}, \eqref{UL:QoS:cons:1}, \nonumber\\
  &\eqref{UL:QoS:cons:2}, \eqref{DL:QoS:cons:1}, \eqref{DL:QoS:cons:2}
\end{align}    
\end{subequations}
Similar to our $\vvFD$ scheme, we find the suboptimal solution to the EE maximization problem of this scheme which is given as
\begin{subequations}
\begin{align} \label{P:EE:RvFD}
	\underset{\THeta, \VARSIGMA, \ALPHA, t, \hat{p}}{\min}\,\, &
	- t
	  \\
	\mathrm{s.t.} \,\,\,\,\,
	& \eqref{theta},
	\eqref{etaa:relation}, \eqref{DL:power:cons}, \eqref{UL:power:cons},
 \eqref{eq:alphml},\eqref{UL:QoS:cons:1}-\eqref{DL:QoS:cons:2},\nonumber\\
 & \eqref{tphat}, \eqref{phat}.
\end{align}
\end{subequations}
Since problems~\eqref{P:SE:RvFD} and~\eqref{P:EE:RvFD} have the same mathematical structure as problems~\eqref{P:SE:equi} and~\eqref{P:EE:fullbh:2}, we solve problem~\eqref{P:SE:RvFD} and problem~\eqref{P:EE:RvFD} by using Algorithms~\ref{alg} and~\ref{alg:EE} with some slight modifications. 

\subsubsection{Network-assisted full-duplex CF-mMIMO with greedy AP mode assignment, fixed power control coefficients and LSFD weights ($\GvFD$)}
The AP mode assignment is performed by a greedy algorithm proposed in~\cite{chowdhury2021can}. 
Let $\mathcal{A}_{\ul}$ and $\mathcal{A}_{\dl}$ denote the sets containing the indices of UL APs, and DL APs, respectively. Also, let $\AAA_s \triangleq \AAA_{\ul} \bigcup \AAA_{\dl}$ be the set of assigned APs and $\AAA_{s'} \triangleq \MM \setminus \AAA_{s}$ be the set of unassigned APs. 
Denote by $\mathcal{S}_{sum}(\mathcal{A}_{\ul}, \mathcal{A}_{\dl}) \triangleq \sum_{k=1}^{K_d}\mathcal{S}_{\ul,k}(\mathcal{A}_{\ul}, \mathcal{A}_{\dl}) + \sum_{\ell=1}^{K_u}\mathcal{S}_{\dl,\ell}(\mathcal{A}_{\ul}, \mathcal{A}_{\dl})$ the sum SE that captures the dependence of the sum SE on the different choices of $\mathcal{A}_{\ul}$ and $\mathcal{A}_{\dl}$. The greedy algorithm for AP mode assignment is shown in Algorithm~\ref{alg:Grreedy}. The key idea of Algorithm~\ref{alg:Grreedy} is to take one AP out of the set of unassigned APs, $\AAA_{s'}$, in each iteration and assign to this AP the mode that offers the highest sum SE until $\AAA_{s'}$ is empty. In this algorithm, the power control coefficients and LSFD weights are fixed, i.e., $\theta_{mk} = \frac{a_{m}}{\sqrt{N K_d \gamdmk}}$, $\varsigma_{\ell} = 1$, $\alpha_{m\ell}=1, \forall m,k,\ell$.

\begin{algorithm}[!t]
\caption{Greedy AP mode assignment for SE maximization~\cite{chowdhury2021can}}
\begin{algorithmic}[1]
\label{alg:Grreedy}
\STATE 
\textbf{Initialize}:  $\mathcal{A}_{\ul}=\mathcal{A}_{\dl}=\emptyset$  
\REPEAT
\STATE  $i_{\ul}^{\star} =\argmax_{i\in\mathcal{A}_{s^{\prime}}} \mathcal{S}_{sum}(\mathcal{A}_{\ul}\bigcup\{i\})$\\
\STATE  $i_{\dl}^{\star} =\argmax_{i\in\mathcal{A}_{s^{\prime}}} \mathcal{S}_{sum}(\mathcal{A}_{\dl}\bigcup\{i\})$
\IF{$\mathcal{S}_{sum}(\mathcal{A}_{\ul}\bigcup\{i_{\ul}^{\star}\}) \geq \mathcal{S}_{sum}(\mathcal{A}_{\dl}\bigcup\{i_{\dl}^{\star}\})$ } 
\STATE {Update  $\mathcal{A}_{\ul}=\mathcal{A}_{\ul}\bigcup\{i_{\ul}^{\star}\}$} 
\ELSE 
\STATE{Update  $\mathcal{A}_{\dl}=\mathcal{A}_{\dl}\bigcup\{i_{\dl}^{\star}\}$} 
\ENDIF
\STATE Update $\mathcal{A}_s=\mathcal{A}_{\ul}\bigcup \mathcal{A}_{\dl}$,
\UNTIL{$\mathcal{A}_{s{^\prime}}=\emptyset$}
\end{algorithmic}
\end{algorithm}

Since Algorithm~\ref{alg:Grreedy} is proposed to only maximize the sum SE, we calculate the EE of this $\GvFD$ scheme by using the AP mode solution obtained from Algorithm~\ref{alg:Grreedy}. Specifically, we use the obtained solution $(\mathcal{A}_{\dl},\mathcal{A}_{\ul})$ to make up the AP mode assignment vectors $(\qa,\qb)$. Then, the EE is calculated by using~\eqref{eq:EE:def} for given $(\qa,\qb)$, in which the total power consumption is computed by~\eqref{eq:Ptotal:final}.   

\subsection{Half-Duplex CF-mMIMO Systems}
To show the advantages of our proposed $\vvFD$ scheme, we compare it with the conventional HD  CF-mMIMO systems ($\HHD$)\cite{Hien:cellfree}. In this system, the DL-and-UL payload data transmission phase is divided into two equal time fractions of length $(\tau_c-\tau_t)/2$. Each UL or DL data transmission is performed in one time fraction. In each time fraction, all UL or DL UEs are served by all the APs, i.e., $a_{m} = b_{m} = 1, \forall m$. There is no interference from the UL UEs to DL UEs, and from the DL APs to UL APs. There is also an additional factor of $\frac{1}{2}$ applied in the SE expression and the total power consumption. This factor captures the fact that each DL or UL transmission is only performed and consumes power in half of the time fraction. In particular, the SE expressions of DL UE $k$ is given by
\begin{align}
\label{SEHDdlk}
&\mathcal{S}_{\dl,k}^{\HD} (\boldsymbol \theta, {\boldsymbol{\varsigma}}) =  \frac{\tau_c-\tau_t}{2\tau_c}
\log_2 \left(1  
    + \text{SINR}_{\dl,k}^{\HD} (\boldsymbol \theta, {\boldsymbol{\varsigma}})
	\right),
\end{align}
where 
$$\text{SINR}_{\dl,k}^{\HD} (\boldsymbol \theta, {\boldsymbol{\varsigma}}) \triangleq \frac{N_t^2 \rho_{d} \left(
                           \sum_{{m\in\mathcal{M}}}
                            \theta_{mk} \gamdmk \right)^2 }{\rho_{d} N_t
    \sum_{k'\in\mathcal{K}_{d}}\sum_{m\in\mathcal{M}}
    \theta_{mk'}^2 \betamkd\gamdmkp +
    1}\vspace{0em},$$ 
while the SE of the UL UE $\ell$ is given by
\begin{align}
    \label{SEHDulell}
    \mathcal{S}_{\ul,\ell}^{\HD} (\boldsymbol{\varsigma}, \boldsymbol{\theta}, \boldsymbol{\alpha} )
	=\! \frac{\tau_c\!-\!\tau_t}{2\tau_c}\log_2
	\left(1+ \text{SINR}_{\ul,\ell}^{\HD}(\boldsymbol{\varsigma}, \boldsymbol{\theta}, \boldsymbol{\alpha} ) \right),
\end{align}
where 
\begin{align}
&\text{SINR}_{\dl,k}^{\HD} (\boldsymbol \theta, {\boldsymbol{\varsigma}}) 
\triangleq
\nonumber\\
&\frac{
	N_r \rho_{u} \left(\sum_{\substack{m\in\mathcal{M}}} \sqrt{b_m \varsigma_{\ell}} \alphml \gamuml \right)^2
	}
	{\!\!\rho_{u}\!
		\! \sum_{\substack{m\in\mathcal{M}}}\!
		\! \sum_{q\in\mathcal{K}_u}\!
		\!\!
		b_m \varsigma_{q}
		\alphml^2
		\betamlu
		\gamma_{m\ell}^{\ul}
		\!+\!
		\!\sum_{\substack{m\in\mathcal{M}}}
		\!
		b_m\alphml^2 \!\gamuml}.\!\vspace{0em}
\end{align}
The total power consumption in the HD scheme is
\begin{align}
    \label{eq:Ptotal:HD}
    P_\mathtt{total}^{\HD} (\x)
    & \triangleq \frac{1}{2} \Bigg[
      \sum_{m\in\mathcal{M}} \frac{N\rho_{d}\Sn}{\zeta_m}\left(\sum_{k\in\mathcal{K}_{d}} \gamdmk \theta_{mk}^2 
    \right) +\PUfix
    \nonumber\\
&\hspace{2em}
+\sum_{\ell\in\mathcal{K}_{u}} \frac{\rho_u\Sn}{\chi} {\varsigma}_{\ell}      
    +\sum_{m\in\MM} \SSS^{\HD} (\x)
    P_{\mathtt{bt},m}
    \nonumber\\
    & \hspace{2em} + \sum_{m\in\mathcal{M}} (N P_{\mathtt{cdl},m} + P_{\mathtt{fdl},m}) 
    \nonumber\\
    &\hspace{2em}
    + \sum_{m\in\mathcal{M}} (NP_{\mathtt{cul},m} + P_{\mathtt{ful},m}) \Bigg],
\end{align}
where $\SSS^{\HD} (\x) \triangleq \sum_{\ell\in\mathcal{K}_u} \mathcal{S}_{\ul,\ell}^{\HD} (\qb, \boldsymbol \varsigma, \boldsymbol \theta, \boldsymbol \alpha)   + \sum_{k\in\mathcal{K}_d}\mathcal{S}_{\dl,k}^{\HD} (\qa, \boldsymbol \theta, {\boldsymbol{\varsigma}}).$

In the $\HHD$ scheme, the power coefficients ($\THeta,\VARSIGMA$) and LSFD weights $\ALPHA$ are optimized to maximize the sum SE and total EE. Note that the mathematical formulas of the SE expressions and the total power consumption in \eqref{SEHDdlk}--\eqref{eq:Ptotal:HD} are the simplified versions of those of \eqref{eq:DL:SE}, \eqref{eq:UL:SE}, and \eqref{eq:Ptotal:upper} of the $\vvFD$ scheme. Therefore, the problems of maximizing the sum SE and EE in the $\HHD$ scheme are similar to problems~\eqref{P:SE:RvFD} and~\eqref{P:EE:RvFD}, and hence, can be solved by using Algorithms~\ref{alg} and~\ref{alg:EE} with appropriate modifications.

\subsection{Full-Duplex CF-mMIMO Systems}
\label{subsec:FD}
We further compare our $\vvFD$ scheme with the traditional FD CF-mMIMO systems ($\FFD$) \cite{tung19ICC,Nguyen:JSAC:2020}. In this scheme, all the APs operate in a FD mode and serve the UL and DL UEs simultaneously at the same frequency. Therefore, $a_m = b_m = 1, \forall m$. Each AP is equipped with $N_t$ transmit antennas and $N_r$ receive antennas. To have a fair comparison with the $\RvFD, \GvFD, \vvFD, \HHD$ schemes, wherein the APs operate in HD mode, the $\FFD$ scheme deploys the same number of antennas as the other schemes, i.e., $N = N_t + N_r$, which is called a ``antenna-number-preserved'' condition \cite{chang12MBC,himal14TWC,mohammad18TWC}.
Therefore, the DL SE of the $\FFD$ scheme is
\begin{align}~\label{eq:DL:SE:FD}
\mathcal{S}_{\dl,k}^{\FD} (\boldsymbol \theta, {\boldsymbol{\varsigma}}) =  \frac{\tau_c-\tau_t}{\tau_c}
\log_2 \left(1  
    + \frac{(\Xi_k^{\FD})^2}{\Omega_k^{\FD}}
	\right),
	\end{align}
where
\begin{align}
    \nonumber
    &\Xi_k^{\FD} (\boldsymbol \theta) \triangleq N_t \sqrt{\rho_{d}}
                           \sum_{{m\in\mathcal{M}}}
                            \theta_{mk} \gamdmk
    \\
    \nonumber
    &\Omega_k^{\FD} (\boldsymbol \theta, {\boldsymbol{\varsigma}}) 
    \!\triangleq \rho_{d} N_t \!
    \!\sum_{k'\in\mathcal{K}_{d}}\!\!\sum_{m\in\mathcal{M}}\!\!
    \theta_{mk'}^2 \betamkd\gamdmkp \!+\! \rho_u\!\sum_{\ell\in\mathcal{K}_u}\!\!  {\vsl} \betakldu \!+\! 1,
\end{align}
while the UL SE of the $\FFD$ scheme is given by
\begin{align}
    \label{eq:UL:SE:FD}
	\mathcal{S}_{\ul,\ell}^{\FD} ( \boldsymbol{\varsigma}, \boldsymbol{\theta}, \boldsymbol{\alpha} )
	=\! \frac{\tau_c\!-\!\tau_t}{\tau_c}\log_2
	(1+ \text{SINR}_{\ul,\ell}^{\FD} (\boldsymbol{\varsigma}, \boldsymbol{\theta}, \boldsymbol{\alpha} )),
\end{align}
where $\text{SINR}_{\ul,\ell}^{\FD}( \boldsymbol{\varsigma}, \boldsymbol{\theta}, \boldsymbol{\alpha} )$ is given at~\eqref{eq:SINRulFD} at the top of the page.
\begin{figure*}
\begin{align}~\label{eq:SINRulFD}
\text{SINR}_{\ul,\ell}^{\FD}( \boldsymbol{\varsigma}, \boldsymbol{\theta}, \boldsymbol{\alpha} )=\frac{
	N_r \rho_{u} \left(\sum_{\substack{m\in\mathcal{M}}} \sqrt{ \varsigma_{\ell}} \alphml \gamuml \right)^2
	}
	{\rho_{u}\!
		\sum_{\substack{m\in\mathcal{M}}}\!
		\sum_{q\in\mathcal{K}_u}\!
		\varsigma_{q}
		\alphml^2
		\betamlu
		\gamma_{m\ell}^{\ul}
		+
		\rho_{d} N_r
		\sum_{\substack{m\in\mathcal{M}}}
		\sum_{\substack{i\in\mathcal{M}}}
		\sum_{k\in\mathcal{K}_d}
		 \theta_{ik}^2 \alphml^2  \gamuml \beta_{mi} \gamma_{ik}^{\dl}
		+
		\sum_{\substack{m\in\mathcal{M}}}
		\alphml^2 \gamuml}.
 \end{align}
  	\hrulefill
	\vspace{-4mm}
  \end{figure*}

Each AP in the FD CF-mMIMO system consumes some amount of power for SIS \cite{Riihonen:TSP:2011,zhang19TWC}.  Let $P_{\text{SIS},m}$ be the power required for SIS at each receive antenna at AP $m$. Then, the total power consumption in the FD CF-mMIMO system  is
\begin{align}
    \label{eq:Ptotal:FD}
    P_\mathtt{total}^{\FD} (\x)
    & \triangleq
      \!\sum_{m\in\mathcal{M}}\!\! \frac{N_t\rho_{d}\Sn}{\zeta_m}\left(\sum_{k\in\mathcal{K}_{d}} \gamdmk \theta_{mk}^2 
    \right) 
    +\!\sum_{\ell\in\mathcal{K}_{u}}\!\!
    \frac{\rho_u\Sn}{\chi} {\varsigma}_{\ell}
    \nonumber\\
    &\hspace{0em}
    +\PUfix 
    +\sum_{m\in\MM} \SSS^{\FD} (\x)
    P_{\mathtt{bt},m}
    \nonumber\\
    & \hspace{0em} + \sum_{m\in\mathcal{M}} (N_t P_{\mathtt{cdl},m} + P_{\mathtt{fdl},m}) 
    \nonumber\\
    &\hspace{0em}
    + \sum_{m\in\mathcal{M}} (N_r P_{\mathtt{cul},m} + P_{\mathtt{ful},m}) + \sum_{m\in\MM} N_r P_{\text{SIS},m},
\end{align}
where 
\begin{align}
    \SSS^{\FD} (\x) \triangleq \sum_{\ell\in\mathcal{K}_u} \mathcal{S}_{\ul,\ell}^{\FD} (\qb, \boldsymbol \varsigma, \boldsymbol \theta, \boldsymbol \alpha)   +  \sum_{k\in\mathcal{K}_d}\mathcal{S}_{\dl,k}^{\FD} (\qa, \boldsymbol \theta, {\boldsymbol{\varsigma}}).
\end{align}
 We note that in a FD CF-mMIMO system, the CPU needs to share with the APs the full signaling loads of all the DL and UL UEs. Therefore, all APs contribute to the two penultimate terms in~\eqref{eq:Ptotal:FD}, while in~\eqref{eq:Ptotal:final}, the APs are partially contributing to the corresponding terms.

We recall that in a $\FFD$ system, the SI can be suppressed using a dedicated hardware and the information of transmit signal. Since the SIS process is normally imperfect, there is still a remaining level of SI, which is called the residual SI after SIS~\cite{Riihonen:TSP:2011}. In a FD CF-mMIMO system, the SI at each AP can be modeled as Rayleigh fading channel~\cite{Riihonen:TSP:2011,hien14JSAC,dan13TSP}. Specifically, we denote by $\Z_{mm}$ the residual SI link at each AP, whose elements are i.i.d $\CN(0,\sigma_{\SSSI}^2)$ RVs, where 
$\sigma_{\SSSI}^2$ is the power of residual SI after SIS at each AP. Note that in our $\vvFD$ scheme where all APs operate in HD mode, there are no SI at each AP, i.e.,  $\sigma_{\SSSI}^2 = 0$. Therefore, the residual SI model is consistent with the model of the interference matrix $\Z_{mm} =0 , \forall m$, in Section~\ref{phase:ULforCE}.

In the $\FFD$ scheme, the power coefficients ($\THeta,\VARSIGMA$) and LSFD weights $\ALPHA$ are also optimized to achieve the maximum sum SE and EE. Therefore, the problems of maximizing the sum SE and EE of the $\FFD$ scheme are similar to problems~\eqref{P:SE:RvFD} and~\eqref{P:EE:RvFD}. Here, the changes in the SE expressions, total power consumption, and the SI matrices make no difference in the mathematical structures of the sum SE and EE maximization problems of the $\FFD$ scheme compared to problems~\eqref{P:SE:RvFD} and~\eqref{P:EE:RvFD}. Thus, the sum SE and EE maximization problems of the $\FFD$ scheme can be solved by using Algorithms~\ref{alg} and~\ref{alg:EE} with some slight modifications.

\begin{remark}
  Our proposed NAFD scheme for CF-mMIMO systems cannot be compared with those in the literature, i.e.,~\cite{Wang:TCOM:2020,Xia:TVT:2020,Jiamin:TWC:2021,Xinjiang:TWC:2021}. These works consider the SE with instantaneous channels, while the SE and EE of our work rely only on the statistical property of the channels. 
\end{remark}

\vspace{0em}
\section{Numerical Examples}~\label{Sec:Numer}

\vspace{-2em}
\subsection{Network Setup and Parameter Setting}
We consider a CF-mMIMO network, where the APs and UEs are randomly distributed in a square of $0.5 \times 0.5$ km${}^2$, whose edges are wrapped around to avoid the boundary effects. The distances between adjacent APs are at least $50$ m\cite{emil20TWC}. Unless otherwise stated, the values of the network parameters are: $N=2$, $N_t=N_r=1$, $\mathcal{S}_\dl^o=\mathcal{S}_\ul^o= \mathcal{S}_{QoS}$ bit/s/Hz, $K_d=K_u =K$, $\tau_c=200$, and $\tau_t=K_d+K_u$. 
We further set the bandwidth $B=50$ MHz and noise figure $F = 9$ dB. Thus, the noise power $\Sn=k_B T_0 B F$, where $k_B=1.381\times 10^{-23}$ Joules/${}^o$K is the Boltzmann constant, while $T_0=290^o$K is the noise temperature. Let $\tilde{\rho}_d = 1$ W, $\tilde{\rho}_u = 0.1$~W and $\tilde{\rho}_t = 0.1$~W be the maximum transmit power of the APs, UL users and UL training pilot sequences, respectively. The normalized maximum transmit powers ${\rho}_d$, ${\rho}_u$, and ${\rho}_t$ are calculated by dividing these powers by the noise power.

\begin{table*}
\vspace{0.5em}
	\caption{Parameters of Power Consumption} 
	\vspace{-0.5em}
	\centering 
	\begin{tabular}{|c | c |}
		\hline
		\textbf{Parameter} & \textbf{Value}  \\ [0.5ex]
		\hline
Fixed power consumption/ each backhaul  ($P_{\mathtt{fdl},m}, P_{\mathtt{ful},m}$, $\forall m$)~\cite{Emil:TWC:2015:EE,Hien:TGCN:2018} & $0.825$ W   \\
		\hline
Internal power consumption/antenna  ($P_{\mathtt{cdl},m}$ and $P_{\mathtt{cul},m}, \forall m$)~\cite{Hien:TGCN:2018} & $0.2$ W   \\
		\hline
Traffic-dependent backhaul power  ($P_{\mathtt{bt},m}, \forall m$)~\cite{Emil:TWC:2015:EE,Hien:TGCN:2018} & $0.25$ W/(Gbits/s) \\
		\hline
Power amplifier efficiency at the APs  ($\zeta_m$, $\forall m$)~\cite{Hien:TGCN:2018} & $0.4$   \\
		\hline
Power amplifier efficiency at the UEs  ($\chi$)~\cite{bashar19TGCN} & $0.3$   \\
		\hline
Fixed power consumption UL and DL UE  ($P_{\mathtt{U},\ell}, P_{\mathtt{D},k}, \forall \ell, k $)~\cite{bashar19TGCN} & $0.1$ W \\
\hline	
	\end{tabular}
	\label{tab:PowerconsumptionParameter}
	\vspace{0.95em}
\end{table*}

We model the large-scale fading coefficients $\beta_{mk}$ as~\cite{emil20TWC}
\begin{align}\label{fading:large}
\beta_{mk} = 10^{\frac{\text{PL}_{mk}^d}{10}}10^{\frac{F_{mk}}{10}},
\end{align}
where $10^{\frac{\text{PL}_{mk}^d}{10}}$ represents the path loss, and $10^{\frac{F_{mk}}{10}}$ represents the shadowing effect with $F_{mk}\in\mathcal{N}(0,4^2)$ (in dB).  Here, $\text{PL}_{mk}^d$ (in dB) is given by  \cite{emil20TWC}
\begin{align}\label{PL:model}
\text{PL}_{mk}^d = -30.5-36.7\log_{10}\left(\frac{d_{mk}}{1\,\text{m}}\right),
\end{align}
and the correlation among the shadowing terms from the AP $m, \forall m\in\mathcal{M}$ to different UEs $k \in\mathcal{K}_d$ ($\ell\in\mathcal{K}_u$) is expressed as:
\begin{align}\label{corr:shadowing}
\mathbb{E}\{F_{mk}F_{jk'}\} \triangleq
\begin{cases}
 4^22^{-\delta_{kk'}/9\,\text{m}},& \text{if $j=m$}\\
 0, & \mbox{otherwise},
\end{cases}, \forall j\in\mathcal{M},
\end{align}
where $\delta_{kk'}$ is the physical distance between UEs $k$ and $k'$.

\begin{figure}[t!]
	\begin{subfigure}[a]{0.5\textwidth}
		\includegraphics[width=90mm]{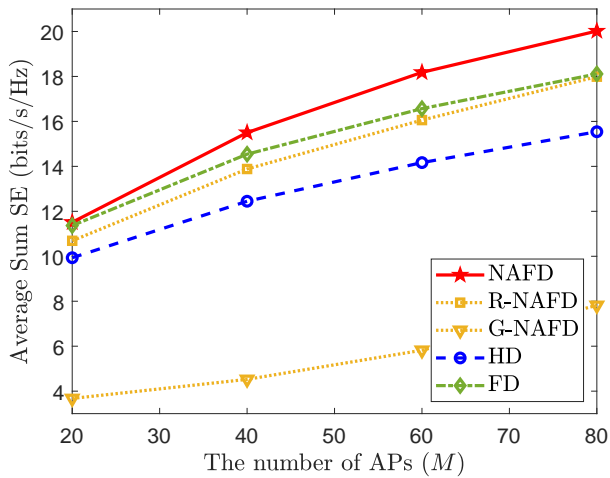}
		\caption{Average sum SE versus the number of APs ($K_d\!=\!K_u\!=\!4$).}
		\label{fig:SEvsM}
	\end{subfigure}
	\begin{subfigure}[a]{0.5\textwidth}
		\includegraphics[width=90mm]{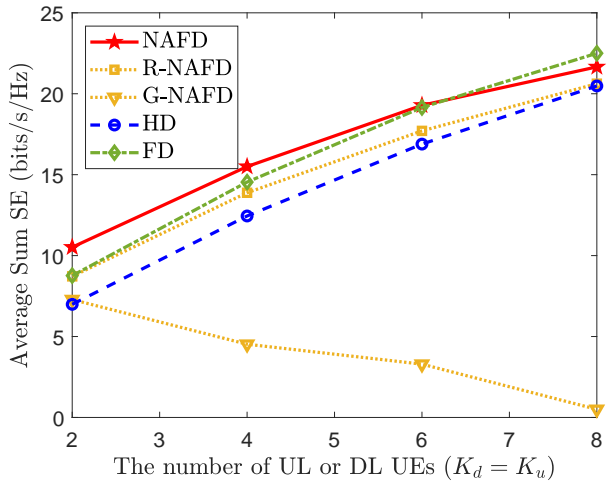}
		\caption{Average sum SE versus the number of UEs ($M=40$).}
		\label{fig:SEvsK}
	\end{subfigure}
	\caption{Comparison among the SE achieved by the proposed $\vvFD$ scheme and baseline schemes ($\mathcal{S}_{QoS}=0.2$ bit/s/Hz, $\sigma_{\text{SI}}^2 /\sigma_n^2 = 50$ dB).}\label{fig:SE-EE}
	\vspace{0em}
\end{figure}

Regarding the power consumption parameters,  we use the parameters of the power consumption in \cite{Emil:TWC:2015:EE,Hien:TGCN:2018,bashar19TGCN}, which are shown in Table~\ref{tab:PowerconsumptionParameter}. We note that the power consumption $P_{\text{SIS},m}$ for SIS strongly depends on the EE of the SIS techniques, used at each AP. Therefore, for a fair comparison, in what follows, we consider the best case of the $\FFD$ scheme with highly energy-efficient SIS techniques, where the power consumption for SIS is sufficiently small and can be ignored in the total power consumption of the $\FFD$ scheme, i.e., $P_{\text{SIS},m} = 0, \forall m$.  

\subsection{Results and Discussions}
We compare our proposed optimized $\vvFD$ scheme with the baseline schemes $\RvFD$, $\GvFD$, $\HHD$, and $\FFD$, which were discussed in Section~\ref{sec:bench}, in terms of sum SE and EE. All the following average results are averaged over $200$ large-scale fading channel realizations. In each channel realization, if the individual SE requirements are not met or the optimization problem of SE or EE maximization of a scheme is infeasible, as discussed in Remarks~\ref{remark:SEfeasible} and~\ref{remark:EEfeasible}, we set the SE or EE of that scheme to zero. The EE values in the results are calculated by using the solution to the problem of maximizing $\EEEE_{fullbh}^{\vFD}$ \eqref{P:EE:fullbh}, which is obtained by Algorithm~\ref{alg:EE}. 

\subsubsection{Effectiveness of the $\vvFD$ scheme in terms of SE} 
Figures~\ref{fig:SEvsM} and~\ref{fig:SEvsK} show the average SE of all the considered schemes versus the number of APs and the different numbers of DL and UL UEs, respectively. Numerical results lead to the following conclusions.

\begin{itemize}
    \item 
    The optimized $\vvFD$ scheme outperforms the heuristic $\RvFD$ and $\GvFD$ schemes. More specifically, it provides performance gains of up to $12\%$ and $150\%$ over $\RvFD$ and $\GvFD$, respectively, which highlights the advantage of our joint optimization solution over heuristic ones.  On the other hand, the remarkable performance gap between the $\RvFD$ and $\GvFD$ verifies the effectiveness of the joint power control and LSFD weight design in NAFD CF-mMIMO systems. Interestingly, $\RvFD$ scheme offers an acceptable SE performance compared with $\vvFD$, which balances the trade-off between performance and complexity. Therefore, it can be deployed instead of the $\vvFD$ scheme if complexity is an issue.  
    \item 
    In the comparison between the proposed $\vvFD$ scheme and the conventional $\HHD$, $\FFD$ schemes, the $\vvFD$ scheme achieves the best SE performance, while the $\HHD$ scheme offers the worst SE performance. This is reasonable because, in $\vvFD$ and $\FFD$ schemes, the UL and DL data transmissions are performed simultaneously, thus, the pre-log factor $\frac{1}{2}$, that comes up in the SE expression of the $\HHD$ scheme, is eliminated. The $\vvFD$ scheme has a smaller number of APs to serve DL or UL UEs than the $\FFD$ scheme, which could lead to lower power for both desired signals and CLI. Nevertheless, with optimizing AP mode assignment, $\vvFD$ is more efficient than $\FFD$ in managing the power resource for sufficiently high power of desired signals and lower power of CLI. Moreover, in the $\vvFD$ scheme, the detrimental impact of residual SI, which exists in the $\FFD$ scheme, is completely removed.
    \item 

    The SE gain of $\vvFD$ over $\HHD$ scheme increases when the number of APs increases. This is because the $\vvFD$ scheme has more degrees-of-freedom in terms of AP mode assignment to manage the CLI.  However, this gain is saturated at around $30\%$. This is because of the AP-to-AP interference. This interference is the fundamental limit of both the $\vvFD$ and $\FFD$ schemes and increases when the number of APs increases. We also note that, due to the presence of residual SI and CLI, the QoS requirement, both the $\vvFD$ and $\FFD$ (under ``antenna-number-preserved'' condition) schemes cannot achieve the promised double gain over the $\HHD$ scheme. When the individual SE constraints are applied, the system needs to spend more power on the UEs with unfavorable links (i.e., lower SEs) to guarantee the SE requirements of these UEs larger than the minimum SE threshold $\SSS_{\ul,\ell}^o$ or $\SSS_{\dl,k}^o$. Therefore, there is stronger interference for the UEs with favorable links, and the SE of these UEs are sacrificed to compensate for the SE of the UEs with unfavorable links.
    \item 
    The SEs of all the considered schemes, except $\GvFD$, are monotonically improved when the number of UEs increases. Note that increasing the number of UEs causes stronger CLI. This result implies that our proposed joint optimization approach can effectively manage the CLI in $\vvFD$, $\RvFD$, $\FFD$ and $\HHD$ schemes. In contrast, the heuristic $\GvFD$ scheme fails to deal with the interference. In our simulation results, it often violates the individual SE requirements and the probability of its SE set to be zero is high.
    \item 
    The gap between the $\vvFD$ and conventional schemes ($\FFD$ and $\HHD$) is diminished when the number of UEs increases. Moreover, $\FFD$ outperforms $\vvFD$ when $K_u=K_d>6$.  These results are reasonable because the number of APs for UL or DL transmission in the $\vvFD$ scheme is smaller than that in the $\FFD$ and $\HHD$ schemes. When the number of UEs increases to a sufficient value, the degrees-of-freedom of the $\FFD$ and $\HHD$ schemes to manage CLI become more than those of the $\vvFD$ and dominate the SE improvement achieved by optimizing the AP mode assignment in $\vvFD$. 
\end{itemize}

\begin{figure}[t!]
	\begin{subfigure}[a]{0.5\textwidth}
		\includegraphics[width=90mm]{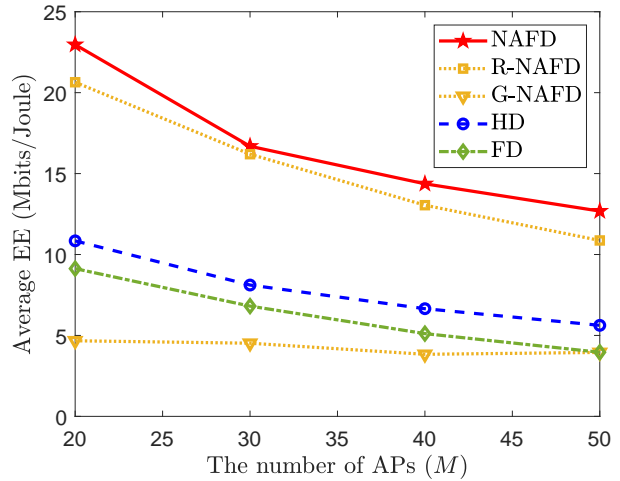}
		\caption{$K_d=K_u=4$.}
		\label{fig:EEvsM}
	\end{subfigure}
	\begin{subfigure}[a]{0.5\textwidth}
		\includegraphics[width=90mm]{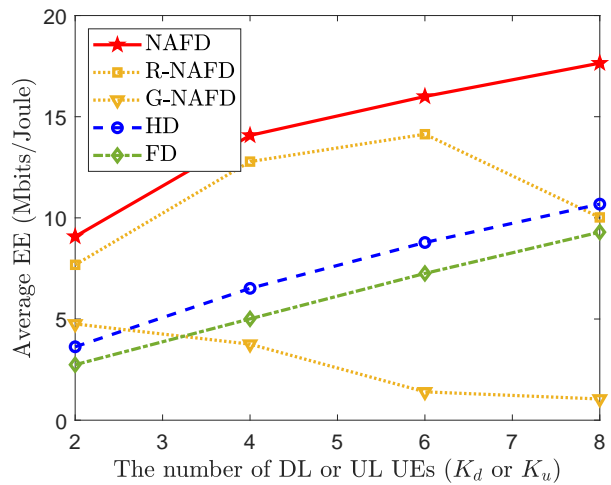}
		\caption{$M=40$.}
		\label{fig:EEvsK}
	\end{subfigure}
	\caption{Comparison among the EE of the proposed $\vvFD$ scheme and baseline schemes ($\mathcal{S}_{QoS}=0.2$ bit/s/Hz, $\sigma_{\text{SI}}^2/\sigma_n^2 = 50$ dB).}\label{fig:EEvK}
	\vspace{-1em}
\end{figure}

\begin{figure}[t!]
	\centering
	\vspace{0em}
	\includegraphics[width=90mm]{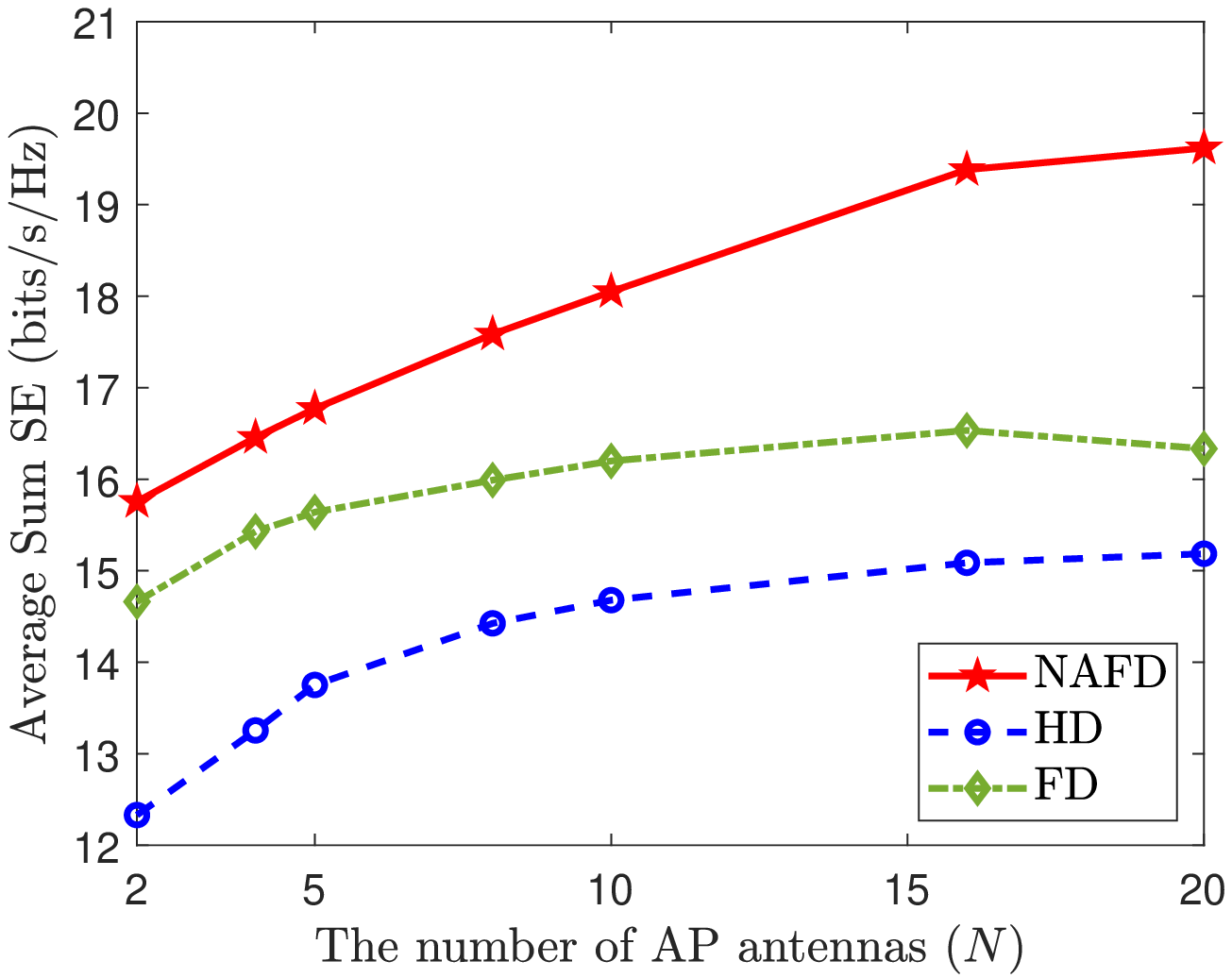}
	\vspace{-0.5em}
	\caption{Impact of the number of AP antennas on the SE ($NM=80, K_d=K_u=4, \mathcal{S}_{QoS} = 0.2$ bit/s/Hz, $\sigma_{\text{SI}}^2/\sigma_n^2 = 50$ dB).}
	\vspace{-1em}
	\label{fig:SEvsN}
\end{figure}

\begin{figure}[t!]
	\centering
	\vspace{0em}
	\includegraphics[width=90mm]{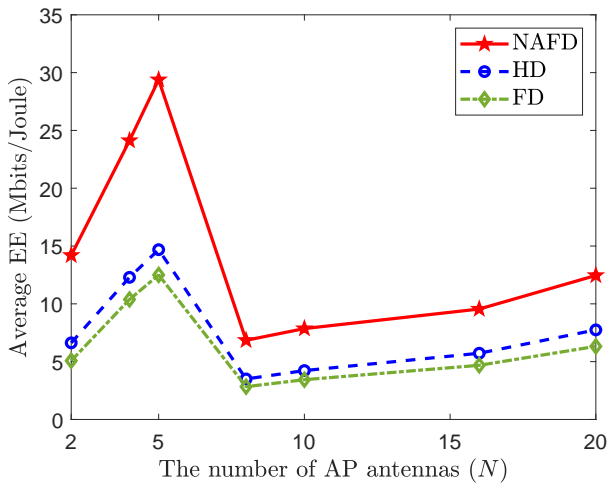}
	\vspace{-0.5em}
	\caption{Impact of the number of AP antennas on the EE ($NM=80, K_d=K_u=4, \mathcal{S}_{QoS} = 0.2$ bit/s/Hz, $\sigma_{\text{SI}}^2/\sigma_n^2 = 50$ dB).}
	\vspace{-1em}
	\label{fig:EEvsN}
\end{figure}

\subsubsection{Effectiveness of the $\vvFD$ scheme in terms of EE} 
Figures~\ref{fig:EEvsM} and~\ref{fig:EEvsK}  show the average EE of the $\vvFD$ and baseline schemes as a function of the number of APs and number of UEs, respectively. From these figures, we have the following observations.
\begin{itemize}
    \item  
    The $\vvFD$ scheme provides a better EE performance than both the heuristic $\RvFD$ and $\GvFD$ schemes. This result highlights the effectiveness of our joint optimization solution over heuristic solutions.  The gap between $\vvFD$ and $\GvFD$ is always noticeable, while that between $\vvFD$ and $\RvFD$ is only remarkable when $K$ is large enough and $M/K$ is small enough. This indicates that the $\RvFD$ scheme can be utilized for the systems that have a large value of $M/K$ with a small number of UEs and requires a lower computational complexity solution with an acceptable loss in the EE performance.
    \item  When the number of UEs increases, the EE performance of $\GvFD$ is degraded, offering the worse performance for $K\geq 4$ among all schemes. However, the EE achieved by the $\RvFD$ is improved by increasing $K$ and then sharply degraded when $K$ increases beyond $6$.
    This is because the CLI is stronger and increases the probability of infeasible cases of these schemes.
    \item 
    The $\vvFD$ scheme achieves a noticeable EE improvement (i.e., up to $200\%$) compared with the $\FFD$ and $\HHD$ schemes. This result is intuitive because the $\vvFD$ scheme can provide a better SE compared with the $\FFD$ and $\HHD$ schemes. Moreover, the $\vvFD$ scheme has a lower number of APs serving DL or UL transmission than the $\RvFD$ and $\GvFD$ schemes, which results in a lower power consumption to the circuit components of the APs (i.e., $P_{\mathtt{cdl},m}$ and $P_{\mathtt{cul},m}$) as well as more power for backhaul links (i.e., $P_{\mathtt{fdl},m}$, $P_{\mathtt{ful},m}$). 
    \item 
    The $\HHD$ scheme slightly outperforms the $\FFD$ scheme in terms of EE. This is due to the fact that in the $\HHD$ scheme, the APs (UEs) consume power in DL (UL) only during half of each time slot, which is reflected by the factor of $\frac{1}{2}$ in~\eqref{eq:Ptotal:HD}. This advantage of the $\HHD$ scheme outperforms the disadvantage of having a lower SE than the $\FFD$ scheme. 
\end{itemize}

\subsubsection{ Impact of the number of AP antennas on SE and EE}
Figures~\ref{fig:SEvsN} and~\ref{fig:EEvsN} illustrate the average sum SE and EE of NAFD, traditional HD, and FD schemes under different numbers of antennas per AP, respectively. Here, the total number of AP antennas is kept fixed, i.e., $NM = 80$. It can be seen from Figure~1 that by increasing the number of antennas at each AP, the sum SE of the NAFD is monotonously increased, while in the cases of FD and HD networks, it does not change much for $N \geq 10$. Therefore, by increasing the number of antennas at each AP, the benefits of the optimized NAFD network are more pronounced. 

From Fig.~\ref{fig:EEvsN}, we can see that there is an optimal number of antennas per AP for a maximum EE, e.g., $N=5$ in the considered setting. In principle, this is  reasonable because when the number of AP antennas increases, the number of APs decreases. On one hand, there are possibly UEs that are now far away from the APs and have low SEs, which slows down the increase in the sum SE. On the other hand, the APs need to use more power to serve these far UEs to guarantee their quality-of-service SEs, which leads to a lower EE. In the regime of small values of $N$, i.e., $N\leq 5$, the number of APs, i.e., $M \geq 16$, is still large and the APs create proper coverage for all the UEs. Thus, increasing the number of APs leads to higher achievable UE rates without increasing the power consumption too much, and hence, increases the EE.  In the regime of large values of $N$, e.g., $N\geq 8$, the number of APs, i.e., $M\leq 10$, is small and the coverage in the network is poor. The increase in power consumption to achieve QoS SEs dominates the increase in the achievable UE rates, which leads to a significant decrease in the EE. Also, in this regime, the EE slightly increases when $N$ increases. This is because the decrease in the number of APs now leads to a reduction in the total fixed power consumption and the total power consumption in backhaul links, while the increase in the achievable SE of UEs is not big enough. 

\subsubsection{Impact of residual SI on the sum SE of the $\FFD$ scheme}
Figure~\ref{fig:SEvsSI} depicts the average SEs of CF-mMIMO systems with the proposed $\vvFD$ and traditional $\HHD$ and $\FFD$ schemes. We recall that the SI has no effect on the performance of the $\vvFD$ and $\HHD$ schemes.  The first observation is that $\FFD$ scheme outperforms $\HHD$ counterpart at a sufficiently small level of residual SI.  $\FFD$ scheme provides more than $18\%$ SE performance gain over the $\HHD$. However, its performance is dramatically degraded when $\sigma_{\text{SI}}^2/\sigma_n^2$ increases. Another interesting observation is that the proposed $\vvFD$ scheme provides the best performance irrespective of the residual SI strength and significantly better performance than $\FFD$ for $\sigma_{\text{SI}}^2/\sigma_n^2>90$ dB. This result further confirms that the $\vvFD$ scheme is well-suited for CF-mMIMO networks. 


\begin{figure}[t!]
	\centering
	\vspace{0em}
	\includegraphics[width=90mm]{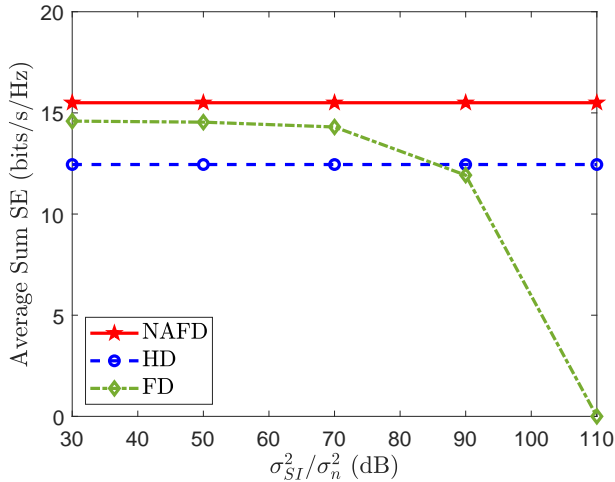}
	\vspace{-0.5em}
	\caption{Average SE versus different residual SI power levels ($M=40, K_d=K_u=4, \mathcal{S}_{QoS} = 0.2$ bit/s/Hz).}
	\vspace{-1em}
	\label{fig:SEvsSI}
\end{figure}

\begin{figure}[t!]
	\centering
	\vspace{0em}
	\includegraphics[width=90mm]{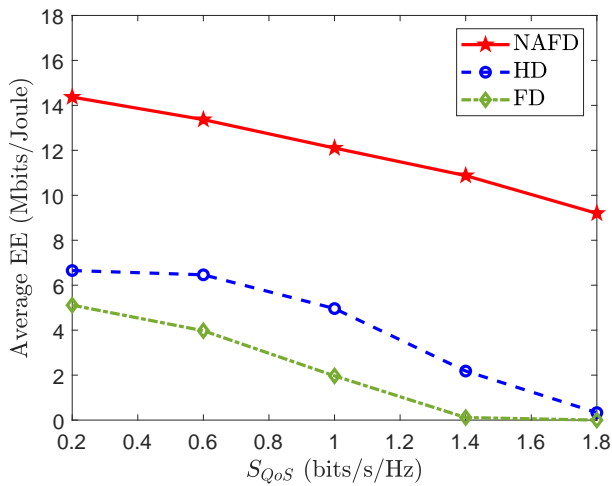}
	\vspace{-0.5em}
	\caption{Impact of individual SE requirements on EE ($M\!=\!40, K_d\!=\!K_u\!=\!4$ bit/s/Hz, $\sigma_{\text{SI}}^2/\sigma_n^2 = 50$ dB).}
	\vspace{-1em}
	\label{fig:EEvsQoS}
\end{figure}

 The SE of the NAFD is relatively greater than that of the FD scheme, especially when the number of APs and/or the number of antennas per each AP is increased. The EE gain achieved by the NAFD over the FD is much greater than the SE gain. This is because the power consumption of the NAFD is remarkably less than that of the FD. In a FD system, all APs contribute to UL and DL transmissions, while in a NAFD system, a part of the APs are assigned for DL transmission and the other part of the APs are used for UL reception. Thus, the power consumption of the transceiver chains and the power consumed for UL and DL transmission is significantly reduced.

\subsubsection{Impact of individual SE requirements on EE}
Figure~\ref{fig:EEvsQoS} shows the EE of CF-mMIMO systems with $\vvFD$, $\HHD$, and $\FFD$ schemes versus the individual SE requirement $\mathcal{S}_{QoS}$. It can be seen that the EE of all schemes decreases as $\mathcal{S}_{QoS}$ increases. This observation can be interpreted as follows. To achieve higher SE requirements, more power is consumed in the system, and then the CLI becomes more severe, which finally leads to lower EE values. Moreover, we can observe that both  $\HHD$ and $\FFD$ schemes fail to satisfy the SE requirement for $\mathcal{S}_{QoS} \geq 1.8$ bits/s/Hz,  while the $\vvFD$ scheme still meets the individual SE requirements of the system. This result shows that the $\vvFD$ scheme can be much more energy-efficient than the $\HHD$ and $\FFD$ schemes while achieving a high SE target.

\subsubsection{Quality of the proposed EE maximization solution}
We investigate the quality of the proposed solution for the EE maximization by demonstrating the gap between $\EEEE^{\vFD}$ and $\EEEE_{fullbh}^{\vFD}$ in Fig.~\ref{fig:EEvsEEfullbh}. It is clear that the gap is small for different numbers of APs $M$. This is because the only difference between maximizing $\EEEE^{\vFD}$ and $\EEEE_{fullbh}^{\vFD}$ is to design the AP mode assignment for optimizing the rate-dependent power consumption $\PbhvFD$ in backhaul links.
However, $\PbhvFD$ contributes only a small portion to the total power consumption $\PtotvFD$, which is shown in Table~\ref{Tab:r}. Specifically, the maximum contribution of $\PbhvFD$ in $\PtotvFD$ is $7.95 \%$, and even the maximum contribution of $\PbhfulvFD$ in $\PtotbhvFD$ is only $14.18\%$. Therefore, optimizing the AP mode assignment for further reducing rate-dependent power consumption $\PbhvFD$ brings no benefit to reduce $\PtotvFD$ much. Hence, the proposed solution  for the $\EEEE_{fullbh}^{\vFD}$ maximization problem in~\eqref{P:EE:fullbh} is sufficiently close to the solution of the $\EEEE$ maximization problem~\eqref{P:EE}. 

\begin{figure}[t!]
	\centering
	\vspace{0em}
	\includegraphics[width=90mm]{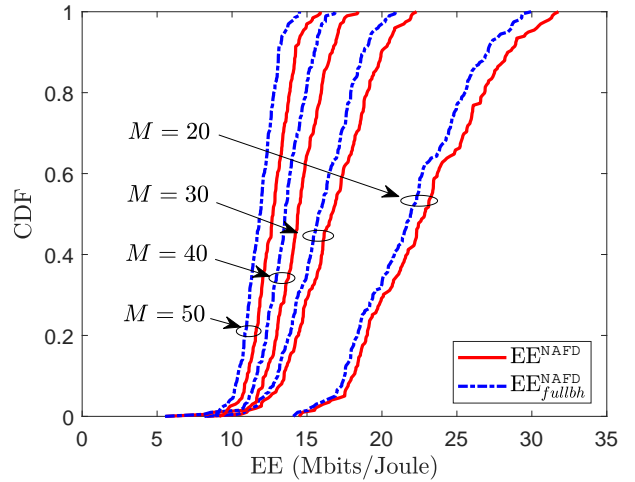}
	\vspace{-0.5em}
	\caption{ Comparison of $\EEEE^{\vFD}$ and $\EEEE_{fullbh}^{\vFD}$ ($K_d=K_u=4, \mathcal{S}_{QoS} = 0.2$ bit/s/Hz, $\sigma_{\text{SI}}^2/\sigma_n^2 = 50$ dB).}
	\vspace{1em}
	\label{fig:EEvsEEfullbh}
\end{figure}

\begin{table}[t]
\renewcommand{\arraystretch}{1.0}
\caption{Average $\frac{\PbhvFD}{\PtotvFD}$ and $\frac{\PbhfulvFD}{\PtotbhvFD}$}
\vspace{1mm}
\label{Tab:r}
\begin{center}
\vspace{-4mm}
\begin{tabular}{|c|c|c|c|c|}
\hline
$M$ & $20$ & $30$ & $40$ & $50$ 
\\
\hline
$\frac{\PbhvFD}{\PtotvFD}$ & $6.71 \%$ & $6.91 \%$ & $7.54 \%$ & $7.95 \%$
\\ &&&&\\
\hline
$\frac{\PbhfulvFD}{\PtotbhvFD}$ & $10.55\%$ & $11.42 \%$ & $12.97 \%$ & $14.18 \%$\\
&&&&\\
\hline
\end{tabular}
\end{center}
\vspace{-1mm}
\end{table}

\section{Conclusion}~\label{Sec:conc}
We have investigated the sum SE and EE performance of NAFD CF-mMIMO systems.  We proposed a large-scale-fading-based joint optimization approach of designing the AP mode assignment, UL and DL power control, and LSFD weights to maximize the sum SE and EE under a realistic power consumption model, individual QoS SE requirements and transmit power constraints. The proposed approach was then applied to maximize the sum SE and EE of CF-mMIMO systems with heuristic NAFD approaches as well as traditional HD and FD approaches. We showed that our jointly optimized NAFD approach provides significant SE and EE gains over the heuristic NAFD approaches. Our results also confirm that in a CF-mMIMO system, the NAFD scheme can achieve a noticeable SE gain, while improving remarkably the  EE compared with the HD and FD schemes. Insights from the simulation results demonstrated that the ratio between the  number of APs and the number of UEs is a dominating factor of the system performance. If the ratio is large, the NAFD scheme with random AP mode assignment offers an acceptable performance; accordingly, balances the performance and complexity of the NAFD scheme. Finally, finding lower-complexity resource allocation approaches, such as machine learning-based algorithms, that can achieve acceptable SE and EE performance in NAFD CF-mMIMO is a timely research topic for future research. 
   
\appendices
\section{ Downlink SE Derivation}
\label{DL:SE:proof}
According to~\eqref{eq:ykdl}, in order to detect $s_{k}^{\dl}$, the $k$-th DL UE need to have access to the effective channel $\sqrt{\rho_d}\sum_{m \in \mathcal{M}} \theta_{mk} \left(\gmkd\right)^T\left(\hgmkd\right)^*$. However, since there is no pilot in DL, this CSI is not available at DL UE $k$. To deal with this challenge, UE $k$ will rely on the stochastic CSI to detect $s_{k}^{\dl}$. The received signal in~~\eqref{eq:ykdl} can be rewritten as
\begin{align}~\label{eq:ykdl:detect}
y_k^{\dl}
&=\mathbb{DS}_k s_{k}^{\dl} + \mathbb{BU}_k s_{k}^{\dl}
+\sum_{k'\in\mathcal{K}_d \setminus k} \mathbb{DI}_{kk'}s_{k'}^{\dl}
\nonumber\\
&\hspace{2em}
+\sum_{\ell\in \mathcal{K}_{u}} \mathbb{UI}_{\ell}s_{\ell}^{\ul}
+w_{k}^{\dl},
\end{align}
where 
\begin{align}
\mathbb{DS}_k &\triangleq\sqrt{\rho_d}\mathbb{E}\bigg\{\sum_{m \in \mathcal{M}} \eta_{mk}^{1/2}
\left(\gmkd\right)^T\left(\hgmkd\right)^*\bigg\},
\nonumber\\
\mathbb{BU}_k &\triangleq\sqrt{\rho_d}
\Bigg( \sum_{m \in \mathcal{M}} \eta_{mk}^{1/2}
\left(\gmkd\right)^T\left(\hgmkd\right)^*
\nonumber\\
&\hspace{4em}
-\mathbb{E}\left\{\sum_{m \in \mathcal{M}} \eta_{mk}^{1/2}
\left(\gmkd\right)^T\left(\hgmkd\right)^*\right\}
\Bigg)\nonumber\\
\mathbb{DI}_{kk'} &\triangleq
\sqrt{\rho_d}
\sum_{m \in \mathcal{M}} \eta_{mk'}^{1/2}
\left(\gmkd\right)^T\left(\hgmkpd\right)^*,
\nonumber\\
\mathbb{UI}_{\ell}&\triangleq
h_{k\ell}\sqrt{\rho_u \tilde{\varsigma}_\ell},
\end{align}
represent the strength of the desired DL signal, the beamforming
gain uncertainty, cross-link interference caused by the $k'$-th DL UE, and cross-link interference caused by the $\ell$-th UL UE, respectively. The sum of the last four terms in~\eqref{eq:ykdl:detect} is treated as the effective noise, which is uncorrelated with the first term, i.e., the desired signal~\cite{Hien:cellfree}. Utilizing the use-and-then-forget capacity bounding technique in~\cite{Hien:cellfree}, the corresponding SE of the  DL UE $k$ is given by
\begin{align}~\label{eq:Sdlk1}
    &\mathcal{S}_{\dl,k}  =  \frac{\tau_c-\tau_t}{\tau_c}
    \log_2 \Big(1 +  \nonumber\\
&\vspace{0em}
\frac{\vert \mathbb{DS}_k \vert^2}
{ \mathbb{E}\left\{\!\vert \mathbb{BU}_k \vert^2 \!\right\}
\!\!+\!\sum_{k'\in\mathcal{K}_d \setminus k}\!\!\mathbb{E}\left\{\!\vert \mathbb{DI}_{kk'} \vert^2\!\right\}   
+\! \sum_{\ell\in \mathcal{K}_{u}}\!\!\mathbb{E}\left\{\!\vert \mathbb{UI}_{\ell} \vert^2\!\right\}
\!+\!1}\Big).
\end{align}

Therefore, we need to compute $\mathbb{DS}_k$, $\mathbb{E}\left\{ \vert \mathbb{BU}_k\vert^2\right\}$, $\mathbb{E}\left\{\vert\mathbb{DI}_{kk'}\vert^2\right\}$, and $\mathbb{E}\left\{\vert \mathbb{UI}_{\ell} \vert^2\right\}$. To compute $\mathbb{DS}_k$, we have
\begin{align}~\label{eq:DSk}
  \mathbb{DS}_k &=\sqrt{\rho_d}\mathbb{E}\bigg\{\sum_{m \in \mathcal{M}} \eta_{mk}^{1/2}
\left(\hgmkd+\tgmkd\right)^T\left(\hgmkd\right)^*\bigg\} 
\nonumber\\
&
=
\sqrt{\rho_d}\sum_{m \in \mathcal{M}} \eta_{mk}^{1/2}
\mathbb{E}\bigg\{\Vert\hgmkd\Vert^2\bigg\}
\nonumber\\ 
&= \sqrt{\rho_d}\sum_{m \in \mathcal{M}} N \eta_{mk}^{1/2} \gamdmk,
\end{align}
where we have used the fact that $\hgmkd$ and $\tgmkd$ are zero mean and independent.

Noticing that  the variance of a sum of
independent RVs is equal to the sum of the variances, we can derive $\mathbb{E}\left\{ \vert \mathbb{BU}_k\vert^2\right\}$ as
\begin{align}~\label{eq:EBUk}
\mathbb{E}\left\{ \vert \mathbb{BU}_k\vert^2\right\} &=
{\rho_d}
\sum_{m \in \mathcal{M}} \eta_{mk}
\mathbb{E}\bigg\{
\bigg\vert
\left(\gmkd\right)^T\left(\hgmkd\right)^*
\nonumber\\
&\hspace{5em}
-
\mathbb{E}\Big\{ 
\left(\gmkd\right)^T\left(\hgmkd\right)^*\Big\}\bigg\vert^2\bigg\}
\nonumber\\
&=
{\rho_d}
\sum_{m \in \mathcal{M}} \eta_{mk}
\bigg(
\mathbb{E}\bigg\{
\bigg\vert
\left(\gmkd\right)^T\left(\hgmkd\right)^*\bigg\vert^2\bigg\}
\nonumber\\
&\hspace{5em}
-
\bigg\vert\mathbb{E}\Big\{ 
\left(\gmkd\right)^T\left(\hgmkd\right)^*\Big\}\bigg\vert^2\bigg)
\nonumber\\
&=
{\rho_d}
\sum_{m \in \mathcal{M}} \eta_{mk}
\bigg(
\mathbb{E}\Big\{ \Big\vert
\left(\tgmkd\right)^T\left(\hgmkd\right)^*\Big\vert^2\Big\}
\nonumber\\
&\hspace{5em}
+
\mathbb{E}\Big\{
\big\Vert\hgmkd\big\Vert^4\Big\}-
N^2(\gamdmk)^2\bigg)
,\nonumber\\
&=
{\rho_d} N
\sum_{m \in \mathcal{M}} \eta_{mk}\gamdmk\betamkd,
\end{align}
where the final result follows from the fact that 
$\mathbb{E}\big\{ \big\vert \left(\tgmkd\right)^T\left(\hgmkd\right)^*\big\vert^2\big\}=\mathbb{E}\big\{ \big\vert \tgmkd\vert^2\big\}\mathbb{E}\big\{\vert\hgmkd\vert^2\big\} = N\gamdmk(\betamkd-\gamdmk)$ and $\mathbb{E}\big\{\big\Vert\hgmkd\big\Vert^4\big\} = N(N+1)(\gamdmk)^2$.

Following similar steps, we can compute $\mathbb{E}\left\{\vert \mathbb{UI}_{\ell} \vert^2\right\}$ and $\mathbb{E}\left\{\vert \mathbb{UI}_{\ell} \vert^2\right\}$ as 
\begin{align}~\label{eq:EDIkkp}
\mathbb{E}\left\{ \vert \mathbb{DI}_{kk'} \vert^2\right\} &={\rho_d} N
\sum_{m \in \mathcal{M}} \eta_{mk}\gamdmk\betamkd,
\nonumber\\
\mathbb{E}\left\{\vert \mathbb{UI}_{\ell} \vert^2\right\} &={\rho_u \tilde{\varsigma}_\ell}\betakldu.
\end{align}

To this end, by substituting~\eqref{eq:DSk}, \eqref{eq:EBUk}, and~\eqref{eq:EDIkkp} into~\eqref{eq:Sdlk1}, the desired result in~\eqref{eq:DL:SE}, is obtained. 


\vspace{-0.4em}
\bibliographystyle{IEEEtran}

\end{document}